\documentclass[a4paper,UKenglish,cleveref, autoref, thm-restate]{lipics-v2021}
\usepackage{algorithm}
\usepackage[noend]{algpseudocode}

\hideLIPIcs  

\newtheorem{thm}[theorem]{Theorem}
\newtheorem{defn}[theorem]{Definition}
\newtheorem{condition}[theorem]{Condition}
\newtheorem{rem}[theorem]{Remark}

\newcommand{\floor}[1]{\left\lfloor #1 \right\rfloor}

\renewcommand{\P}{\mathbf{P}}
\newcommand{\rand}{\overset{\textbf{R}}\leftarrow}

\DeclareMathOperator{\TV}{\mathsf{TV}}

\newcommand{\EE}{\mathbb{E}}

\newcommand{\RR}{\mathbb{R}}
\newcommand{\PP}{\mathbb{P}}
\newcommand{\NN}{\mathbb{N}}

\newcommand{\ZZ}{\mathbb{Z}}

\newcommand{\mcC}{\mathcal{C}}

\newcommand{\mcM}{\mathcal{M}}

\newcommand{\cC}{\mathcal{C}}

\newcommand{\eps}{\varepsilon}

\newcommand{\MS}{\mathsf{MarginSample}}
\newcommand{\GB}{\mathsf{Glauber}}
\newcommand{\IM}{\mathsf{IsMarked}}
\newcommand{\US}{\mathsf{UniformSample}}
\newcommand{\CONN}{\mathsf{Conn}}
\newcommand{\TS}{\phi}
\newcommand{\LB}{\textsf{LB-Sample}}
\newcommand{\pad}{\mathrm{pad}}
\newcommand{\pred}{\mathsf{pred}}
\newcommand{\lb}{\mathrm{LB}}

\DeclareMathOperator{\vbl}{vbl}

\newcommand{\poly}{\mathrm{poly}}
\newcommand{\dist}{\mathrm{dist}}

\renewcommand{\epsilon}{\varepsilon}
\renewcommand{\subset}{\subseteq}

\newcommand{\mcA}{\mathcal{A}}

\newcommand{\mcE}{\mathcal{E}}

\newcommand{\mcS}{\mathcal{S}}

\renewcommand{\P}{\mathbb{P}}

\title{Random local access for sampling $k$-\textsf{SAT} solutions} 

\author{Dingding Dong}{Harvard University, Cambridge MA, USA\\ddong@math.harvard.edu}{ddong@math.harvard.edu}{https://orcid.org/0000-0001-8500-2897
}{}

\author{Nitya Mani}{Massachusetts Institute of Technology, Cambridge MA, USA\\nmani@mit.edu}{nmani@mit.edu}{https://orcid.org/0000-0001-8500-2897
}{}

\authorrunning{Dingding Dong and Nitya Mani} 

\Copyright{Dingding Dong and Nitya Mani} 



\category{} 

\relatedversion{} 

\funding{NM was supported by the Hertz Graduate Fellowship and the NSF Graduate Research Fellowship Program \#2141064. This material includes work supported by the NSF Grant DMS-1928930, while the authors were in residence at SLMath during the Spring 2025 semester.}

\acknowledgements{We would like to thank Kuikui Liu and Ronitt Rubinfeld for helpful discussions, suggestions, and feedback on a draft. We would like to thank the anonymous referees for identifying an issue with a previous version of this article.}

\nolinenumbers 

\EventEditors{Jeremias Berg and Jakob Nordstr\"{o}m}
\EventNoEds{2}
\EventLongTitle{28th International Conference on Theory and Applications of Satisfiability Testing (SAT 2025)}
\EventShortTitle{SAT 2025}
\EventAcronym{SAT}
\EventYear{2025}
\EventDate{August 12--15, 2025}
\EventLocation{Glasgow, Scotland}
\EventLogo{}
\SeriesVolume{341}
\ArticleNo{10}

\begin{document}

\maketitle

\begin{abstract}
We present a \textit{sublinear time} algorithm that gives \textit{random local access} to the uniform distribution over satisfying assignments to an arbitrary $k$-$\textsf{SAT}$ formula $\Phi$, at exponential clause density. Our algorithm provides \textit{memory-less} query access to variable assignments, such that the output variable assignments consistently emulate a single global satisfying assignment whose law is close to the uniform distribution over satisfying assignments to $\Phi$. Random local access and related models have been studied for a wide variety of natural Gibbs distributions and random graphical processes. Here, we establish feasibility of random local access models for one of the most canonical such sample spaces, the set of satisfying assignments to a $k$-$\textsf{SAT}$ formula.

Our algorithm proceeds by leveraging the local uniformity of the uniform distribution over satisfying assignments to $\Phi$. We randomly partition the variables into two subsets, so that each clause has sufficiently many variables from each set to preserve local uniformity. We then sample some variables by simulating a systematic scan Glauber dynamics backward in time, greedily constructing the necessary intermediate steps. We sample the other variables by first conducting a search for a polylogarithmic-sized local component, which we iteratively grow to identify a small subformula from which we can efficiently sample using the appropriate marginal distribution. This two-pronged approach enables us to sample individual variable assignments without constructing a full solution.

\end{abstract}

\section{Introduction}
Efficiently sampling from an exponentially sized solution space is a fundamental problem in computation. The quintessential such problem is sampling a uniformly random satisfying assignment to a given $k$-\textsf{SAT} formula, the setting of this work. Throughout, we let $\Phi$ denote an $n$-variable \textit{$(k, d)$-formula}, which is a $k$-\textsf{SAT} formula in \textit{conjunctive normal form} (which we also call a $k$-CNF) such that each variable appears in at most $d$ clauses. We let $\Omega = \Omega_{\Phi}$ denote the set of satisfying assignments to $\Phi$, and let $\mu = \mu_{\Phi}$ denote the uniform distribution over $\Omega$. We consider the regime where $k$ is a (large) constant and take $n \to \infty$.

For $d \le (2ek)^{-1} 2^{k}$, it has long been known (e.g., by the Lov\'asz local lemma introduced in~\cite{EL75}) that any such $k$-\textsf{SAT} formula $\Phi$ has at least one satisfying assignment. Moreover, seminal work of Moser and Tardos~\cite{MT10} gave a simple algorithm that computes such a satisfying assignment in expected linear time. However, a similarly basic question in the sampling regime is far more poorly understood: when is it possible to efficiently sample an (approximately uniformly) random satisfying assignment to $\Phi$? 

The first breakthrough in this direction came from Moitra~\cite{Moi19} who gave a deterministic, fixed-parameter tractable algorithm for approximately sampling from $\mu$ provided that $d \lesssim 2^{ck}$ for $c \approx 1/60$. This was followed up by work of Feng, Guo, Yin and Zhang~\cite{FGY21}, who used a Markov chain approach to give a polynomial time algorithm for sampling from $\mu$ provided that $d \lesssim 2^{k/20}$. Recent works~\cite{HWY22,WY24} have given state of the art results for sampling from $k$-CNFs, when $d \lessapprox 2^{k/4.82}$ (in fact for large domain sizes,~\cite{WY24} gives essentially tight bounds on sampling from atomic constraint satisfaction problems). There has also been a tremendous amount of work sampling from random $k$-CNFs, recently essentially resolved by work of~\cite{CLWYY25}, building on previous progress of~\cite{GGGY21,CMM23,GGGH22+,HE23}.

A natural question about such algorithms to sample from $k$-CNFs is whether one can adapt them to more efficiently answer \textit{local} questions about a random satisfying assignment in \textit{sublinear time}. For variable $v$, let $\mu_v$ denote the marginal distribution on $v$ induced by $\mu$. When $n$ is large, one might wish to sample a small handful of variables $S$ from their marginal distributions $\mu_S$ in $o(n)$ time, without computing an entire $\Omega(n)$-sized satisfying assignment $\sigma\sim \mu$. 
Further, many traditional algorithms for counting the number of satisfying assignments to a given $k$-\textsf{SAT} formula proceed by computing marginal probabilities of variable assignments, a task that can be completed given local access to a random assignment. Therefore, sublinear time algorithms for answering local questions can also yield speedups in more general counting algorithms.

Ideally, a \textit{random local access} model should provide query access to variable assignments such that the output enjoys the following properties: 
\begin{enumerate}[(a)]
    \item the model is \textit{consistent}: queries for different variable assignments consistently emulate those of a single random satisfying assignment $\sigma \sim \mu$;
    \item the model is \textit{sublinear}: sampling variable assignments takes $o(n)$ time in expectation;
    \item the model is \textit{memory-less}: given the same initialization of randomness, answers to  multiple, possibly independent queries for variable assignments are consistent with each other.
\end{enumerate}

We give a more precise definition of random local access models in~\cref{s:prelim}. 
Random local access models were formally introduced in work of Biswas, Rubinfeld, and Yodpinyanee~\cite{BIS17}, motivated by a long line of work over the past decades studying sublinear time algorithms for problems in theoretical computer science. The authors of~\cite{BIS17} studied a different natural Gibbs distribution, the uniform distribution over proper $q$-colorings of a given $n$-vertex graph with maximum degree $\Delta$. By adapting previously studied classical and parallel sampling algorithms for $q$-colorings, they were able to construct a local access algorithm to a random proper $q$-coloring when $q \ge 9\Delta.$ The related problem of sampling partial information about huge random objects was pioneered much earlier in \cite{GGM86,GGN10}; further work in \cite{NN07} considers the implementation of different random graph models. Random local access and related (e.g., parallel computation) models have also been recently studied for several other random graphical processes and Gibbs distributions (cf.~\cite{BIS17,BIS21,MSW22,LY22}). 

The above model for local access to random samples extends a rich line of work studying \textit{local computation algorithms} (LCAs), originally introduced in works of Rubinfield, Tamir, Vardi, and Xie \cite{RTVX11} and Alon, Rubinfield, Vardi, and Xie~\cite{ARV12}. Local computation algorithms are widely used in parallel and distributed computing, with notable practical success in areas such as graph sparsification \cite{AGM12}, load balancing \cite{MRVX12}, and sublinear-time coloring \cite{ARV12}. A primary initial application of LCAs in~\cite{RTVX11} was locally constructing a satisfying assignment to a given $(k, d)$-formula $\Phi$ when $d \lesssim 2^{k/10}$. Similar to the desired properties of random local access, LCA implements
query access to a large solution space in sublinear time using local probes. However, rather than sampling from a given distribution, LCA only provides local access to \textit{some} consistent valid solution in the desired solution space that may otherwise be completely arbitrary in nature. 

In this work, we present a polylogarithmic time algorithm that gives random local access to the uniform distribution over satisfying assignments of an arbitrary $k$-\textsf{SAT} formula $\Phi$ at exponential clause density (i.e., the number of occurrences of each variable is bounded by an exponential function of $k$). The following is an informal version of our main result; a more formal version can be found at~\cref{t:main-formal}.

\begin{thm}[Main theorem: informal]\label{t:main}
    There exists a random local access algorithm $\mcA$ that satisfies the following conditions. For any $(k,d)$-formula $\Phi=(V,\mcC)$ with maximum degree $d \lesssim 2^{k/400}$ and any variable $v\in V$, we can use $\mcA$ to sample a single assignment for $v$, $\mcA(v)\in\{0,1\}$ in expected polylogarithmic time (in $|V|$), such that the distribution of $\mcA(v)$ is the same as the marginal distribution of $\sigma(v)$, where $\sigma$ is a uniformly random satisfying assignment to $\Phi$.
\end{thm}


Similar to~\cite{BIS17}, the proof of~\cref{t:main} adapts some of the algorithmic tools used to study parallel and distributed sampling. The proof also builds upon the work of Moitra \cite{Moi19} and Feng--Guo--Yin--Zhang \cite{FGY21} on sampling from bounded degree $k$-CNFs in polynomial time.
The authors of \cite{Moi19, FGY21} both critically took advantage of a global variable \textit{marking} (see \cref{def:marking}) to better condition the marginal distributions of variables. Such an approach allows for a subset of the variable assignments to be sampled with ease; the resulting \textit{shattering} of the solution space conditioned on such a partial assignment then allows one to efficiently complete the random satisfying assignment. These initial approaches have been extended and strengthened to nearly linear time algorithms that succeed for larger maximum variable degree in a flurry of recent work (c.f.~\cite{Moi19,FGY21,JPV21,HE21,GGGY21,CMM23,GGGH22+,WY24,CLWYY25}). 

Recently, Feng, Guo, Wang, Wang, and Yin~\cite{FGWWY24} used a recursive sampling scheme to simulate the systematic scan projected Glauber dynamics via a strategy termed \textit{coupling towards the past}, which they used to derandomise several Markov chain Monte Carlo (MCMC) algorithms for sampling from CSPs. Additionally, recent work of He, Wang and Yin~\cite{HWY22} used a \textit{recursive sampling} scheme to sample $k$-\textsf{SAT} solutions. Their work immediately gives a sublinear (in fact, expected constant time) algorithm for sampling the assignment of a single variable one time; however, this work does not immediately extend out of the box to a random local access model that enjoys the desired \textit{consistency} and \textit{memory-less} properties if multiple variables are sampled.
Recursive sampling schemes have also been recently used to make progress on analyzing and designing fast sampling algorithms for a variety of Gibbs distributions (cf.~\cite{HWY22,CMM23,AJ22}). As noted earlier, such schemes have been particularly useful for \textit{derandomising} and constructing \textit{parallel} and \textit{distributed} versions of many popular MCMC algorithms for sampling solutions to CSPs~\cite{FGWWY24,FY18,HE21,HWY22}.

An immediate roadblock to adapting such global or parallel sampling strategies to random local access to $k$-\textsf{SAT} solutions is that the vast majority of the aforementioned existing algorithms critically either use some \textit{global} information -- like a variable marking, particular ordering of variables, or other global state compression -- as an input to make algorithmic decisions or postpone sampling of certain problematic variable assignments until a \textit{linear} number of other choices are made. Both these issues necessitate a substantive departure from these approaches for any hope of local access adaptation. In this work, we overcome these obstacles by adapting the \textit{coupling towards the past} strategy used in~\cite{FGWWY24} to derandomise MCMC in conjunction with a local implementation of the variable marking strategy introduced in~\cite{Moi19}.

We use these algorithms to carefully select and sample a small number of auxiliary variable assignments on an as-needed basis, showing that bad cases can be reduced to small calculations after a sublinear number of steps. Our proof of~\cref{t:main} requires a novel adaptation of sampling techniques to avoid requiring global information or complete passes of the variable set; we show that we can perform analogous operations locally or leverage local uniformity properties to circumvent them altogether in the course of locally sampling a variable assignment. Importantly, we demonstrate how to execute local sampling in an \textit{oblivious, consistent} fashion, so that the algorithm need not retain any memory and that variables need not be sampled in any particular order.

\section{Preliminaries}\label{s:prelim}

\subsection*{Notation}
Throughout, let $\Phi = (V, \mcC)$ be a $k$-CNF on variable set $V = \{v_1, \ldots, v_n\}$ with associated collection of clauses $\mcC$. In this work we do not let the same variable appear multiple times in any clause of $\Phi$, although our algorithm could be adapted to the more general scenario. We further assume that $\Phi$ is a \textit{$(k, d)$-formula}; that is, each variable $v_i$ appears in at most $d$ clauses. For every clause $C \in \mcC$, let $\vbl(C)$ denote the collection of variables in the clause $C$. We further define $\Delta:=\max_{C\in \mcC}|\{C'\in \mcC:\vbl(C)\cap \vbl(C')\neq\emptyset\}|$, so in particular $\Delta\leq kd$.

In the regime we work in, we assume $k$ is a large fixed constant and view $n \to \infty.$ We use $f \lesssim g$ to denote that  there is some constant $C$ (not depending on $k$) such that $f \le Cg$.
 We also use the standard asymptotic notation ($O$, $o$, $\Omega$, $\omega$, $\Theta$), where when not specified, we assume these are in the $n \to \infty$ limit. We use $\mathsf{Law}(X)$ to denote the underlying distribution of a random variable $X$.

We let $\Omega = \Omega_{\Phi} \subset \{0, 1\}^n$ denote the set of satisfying assignments to $\Phi$ and let $\mu = \mu_{\Phi}$ denote the uniform distribution over $\Omega$. We suppress the $\Phi$ subscripts when the formula is clear from context. We introduce a few more concepts associated to a $k$-\textsf{SAT} instance that will be used later.

\begin{defn}
 Given probability distributions $\nu_1, \nu_2$ over $\Omega$, the \textit{total variation distance} is
$$d_{\TV}(\nu_1, \nu_2) := \frac12 \sum_{\omega \in \Omega} |\nu_1(\omega) - \nu_2(\omega)|.$$   
\end{defn}
If we have a random variable $X$ with $\textsf{Law}(X) = \nu$, we may write $d_{\TV}(\mu, X)$ instead of $d_{\TV}(\mu, \nu)$ in a slight abuse of notation.
 
\begin{defn}[Dependency hypergraph]\label{def:dependency-graph}
Given a $(k, d)$-formula $\Phi$, let $H_{\Phi} = (V, \mcE)$ be the \textit{dependency hypergraph} (with multiple edges allowed), where $V$ is the set of variables and $\mcE = \{ \vbl(C) : C \in \mcC\}$ is the collection of variable sets of clauses of $\Phi$.
\end{defn}

\begin{defn}[Partial assignment]\label{def:partial-assignment}
    Given $(k, d)$-formula $\Phi=(V,\mcC)$, let
    $$
    \mathcal Q^*:=\bigcup_{\Lambda\subseteq V}\{0,1,\perp\}^\Lambda
    $$
    denote the space of partial assignments with the following symbology. Given a partial assignment $\sigma\in\{0,1,\perp\}^\Lambda$ on some  $\Lambda\subseteq V$, each variable $v\in \Lambda$ is classified as follows:
    \begin{itemize}
        \item $\sigma(v)\in\{0,1\}$ means that $v$ is \textit{accessed} by the algorithm and \textit{assigned} with  $\sigma(v)\in\{0,1\}$;
        \item $\sigma(v)=\perp$ means that $v$ is \textit{accessed} by the algorithm but \textit{unassigned} yet with any value.
    \end{itemize}
We sometimes use $\emptyset$ to denote the empty assignment (i.e., $\Lambda$ is the empty set).
We say that $\sigma$ is \textit{feasible} if it can be extended to a satisfying assignment to $\Phi$. 
\end{defn}

\begin{defn}[Reduced formula on partial assignment]\label{def:reduced}
Let $\Phi$ be a $(k, d)$-formula.
Given a partial assignment $\sigma$ on $\Lambda\subseteq V$, let $\Phi^\sigma$ be the result of simplifying $\Phi$ under $\sigma$. That is, define $\Phi^\sigma := (V^\sigma, \mcC^\sigma)$ where
\begin{itemize}
    \item $V^\sigma = V \setminus \{v\in \Lambda:\sigma(v)\in\{0,1\}\}$,
    \item $\mcC^\sigma$ is obtained from $\mcC$ by removing all clauses that have been satisfied under $\sigma$, and removing any appearance of variables that are assigned 0 or 1 by $\sigma$.
\end{itemize}
Let $H_{\Phi}^\sigma$ be the associated (not necessarily $k$-uniform) hypergraph to $\Phi^\sigma$. For variable $v \in V\setminus \Lambda$, let $\Phi_v^\sigma$ denote the maximal connected component of $\Phi^\sigma$ to which $v$ belongs.
\end{defn}

\begin{defn}[Marginal distribution]
For an arbitrary set of variables $S \subset V$, let $\mu_S$ be the marginal distribution on $S$ induced by $\mu$, so that
$$ \mu_S(\tau) = \sum_{\tau' \in \{0, 1\}^n : \tau'|_S = \tau} \mu(\tau') \qquad \forall \tau \in \{0, 1\}^S.$$
When $S = \{v \}$ is a single vertex, we let $\mu_v = \mu_{\{v\}}$. Furthermore, given some partial assignment $\sigma \in \{0, 1,\perp\}^{\Lambda}$ for $\Lambda \subset V$, if $S \cap \Lambda = \emptyset$, we let $\mu_{S}^\sigma(\cdot) := \mu_S(\cdot \mid \sigma)$ be the marginal distribution on $S$ conditioned on the partial assignment $\sigma$. When $S=V\setminus \Lambda$, we simply use $\mu^\sigma$ to denote $\mu^\sigma_{V\setminus \Lambda}$.
\end{defn}

\subsection{The random local access model and local computation algorithms}

One of the most widely studied models of local computation are local computation algorithms (LCAs) introduced in~\cite{ARV12,RTVX11}. Given a computation problem $F$, an LCA (in an \textit{online fashion}) provides the $i$-th bit to \textit{some} consistent solution $F$ in sublinear time. As originally defined, local computation algorithms need not be \textit{query-oblivious}; in other words, the output solution can depend on the order of queried bits. However, several follow-up works have given query-oblivious analogues of LCAs for a variety of natural problems. Such models are the non-random-sampling version of the random local access models we study here.

In this work, we construct a query-oblivious LCA for an intermediate step in our construction of the random local access model (described in more detail in~\cref{s:marking}). We thus precisely define both LCAs and random local access models below.

\begin{defn}[Local computation algorithm]
Given an object family $(\Pi, \Omega )$ with input $\Pi$ and sample space $\Omega \subset \Sigma^n$ (for some alphabet $\Sigma$),  a \textit{$(t(n), \delta(n))$-local computation algorithm} (LCA) provides an oracle $\mcA$ that implements query access to some arbitrary $\sigma \in \Omega$ that satisfies the following conditions:
\begin{itemize}
\item $\mcA$ has query access to the input description $\Pi$ and to a tape of public random bits $\mathbf{R}$.
\item $\mcA$ gets a sequence of queries $i_1, \ldots, i_q$ for any $q > 0$, and after
each query $i_j$, it produces an output $\sigma_{i_j}$ such that the outputs $\sigma_{i_1}, \ldots, \sigma_{i_q}$ are consistent with some $\sigma \in \Omega$.
\item The probability of success over all $q$ queries is at least $1 - \delta(n)$ (where $\delta(n) < 1/3$).
\item  The expected running time of $\mcA$ on any query is at most $t(n)$, which is sublinear in $n$.
\end{itemize}
We further say that $\mcA$ is \textit{query-oblivious} if the outputs of $\mcA$ do not depend on the order of the queries but depend only on $\Pi$ and $\mathbf{R}$.
\end{defn}

Motivated by the above, we give a formal definition of the random local access model introduced in~\cite{BIS17}.

\begin{defn}[Random local access]
Given a random object family $(\Pi, \Omega, \mu)$ with input $\Pi$, sample space $\Omega \subset \Sigma^n$ (with alphabet $\Sigma$) and distribution $\mu$ supported on $\Omega$, a \textit{$(t(n), \delta(n))$-random local access implementation} of a family of query functions $\{F_1, F_2, \ldots \}$ provides an oracle $\mcA$ with the following properties:
\begin{itemize}
    \item $\mcA$ has query access to the input description $\Pi$ and to a tape of public random bits $\mathbf{R}.$
    \item For a given input $\Pi$, upon being queried with $F_i$, the oracle with probability $1 - o(1)$ uses at most $t(n)$ resources (where $t(n)$ is sublinear in $n$)  to return the value $\mcA(\Pi, \mathbf{R}, F_i(Y))$ for some specific $Y \in \Omega$.
    \item The choice of $Y \in \Omega$ only depends on $\Pi$ and $\mathbf{R}$.
    \item The distribution of $Y$ over the randomness $\mathbf{R}$ satisfies 
    $$d_{\TV} (Y, \mu)  = \frac12\sum_{\omega \in \Omega} |\P(Y = \omega) - \mu(\omega)| < \delta(n),$$
    where $\delta(n) \lesssim \frac{1}{n^c}$ for constant $c > 1$.
\end{itemize}
\end{defn}
In other words, if $\mcA$ is a random local access oracle for a set of queries, then when provided the same input $\Pi$ and the same random bits $\mathbf{R}$, it must provide outputs that are consistent with a single choice of $Y$ regardless of the order and content of these queries.

\begin{rem}
In this work, we do not study or seek to optimize the memory usage of our algorithms. However, there is also a rich literature studying space-efficient and parallelizable local models (e.g.,~\cite{ARV12,GHA19,KMW16}).
\end{rem}

\subsection{Marking and a query-oblivious LCA}
\label{s:marking}

As noted in the introduction, the Lov\'asz local lemma \cite{Spe77} guarantees the existence of a satisfying assignment to any $(k, d)$-formula $\Phi$ if $d \le (2ek)^{-1} 2^k.$ Furthermore, the Moser-Tardos algorithm~\cite{MT10} gives a simple linear-time algorithm to construct such a satisfying assignment in the above $(k, d)$-regime:
\begin{itemize}
    \item Sample $v_1, \ldots, v_n \rand \{0, 1\}$ uniformly at random;
    \item While there exists a clause $C \in \mcC$ that is violated by the current assignment, resample variables in $C$ uniformly at random.
\end{itemize}

The Lov\'asz local lemma can be made \textit{quantitative}, showing that not only is there some satisfying assignment to $\Phi$ if $d \le (2ek)^{-1} 2^k$, but both that there are exponentially many such satisfying assignments and that the marginal distributions $\mu_v$ are approximately uniform with $\mu_v(1) \le \frac12 \exp(1/k)$ (see~\cite{Moi19,FGY21}). Such \textit{local uniformity} is critical to the success of algorithms that approximately sample from $\mu$. In his breakthrough work, Moitra~\cite{Moi19} noted that this local uniformity continues to hold for conditional distributions $\mu^\sigma$ provided that each clause has a sufficiently large number of free variables under partial assignment $\sigma$. This motivates the idea of a \textit{marking}, as introduced in~\cite{Moi19}, which is a careful selection (via the Lov\'asz local lemma) of a linear sized subset of variables $\mcM \subset V$ that has the following properties: 
\begin{itemize}
    \item For every clause $C$, $|\vbl(C)\cap \mcM| \gtrsim k$. Having a large number of marked variables in each clause would hopefully result in a desired \textit{shattering condition}; namely, we can sample a partial assignment $\sigma\in  \{0, 1\}^{\mcM }$ on this marking that partitions the original formula into sufficiently small connected components.
    
    \item For every clause $C$, $|\vbl(C) \backslash \mcM| \gtrsim k$. This large non-intersection preserves the \textit{local uniformity} of the marginal distributions of the as yet-unsampled variables in $\mcM$.
\end{itemize}

The general strategy of arguments following the marking approach is to show that it is ``easy'' to sample a partial assignment $\sigma\in  \{0, 1\}^{\mcM }$, and moreover, conditioned on any such assignment, the reduced formula $\Phi^\sigma$ is very likely to shatter into sufficiently small connected components such that the remaining variable assignments can be efficiently sampled from the conditional distribution by solving a smaller instance of the original problem. We now make this notion precise.

\begin{defn}[$\alpha$-marking]\label{def:marking}
Given $(k, d)$-formula $\Phi$ and constant $\alpha \in (0, 1/2)$, we say that $\mcM \subset V$ is an $\alpha$-marking if for every $C \in \mcC$, $|\vbl(C)\cap \mcM | \ge \alpha k$ and $|\vbl(C) \backslash \mcM| \ge \alpha k$.
\end{defn}

In this work, we locally, greedily construct an $\alpha$-marking $\mathcal M$ using a local implementation of Moser-Tardos. We will further argue that because of the shattering property, we can locally compute the connected component of $\Phi_{v}^\sigma$ for some $\sigma\sim\mu_{\mathcal M}$ and a given vertex $v$, without having to actually assign linearly many vertices.

Precisely, we construct a query-oblivious LCA, $\mathsf{IsMarked}(\cdot)$, where for a $(k, d)$-formula $\Phi$ and a given range of $\alpha \in (0, 1)$, $\mathsf{IsMarked}(\cdot)$ can take in any variable $v\in V$ and output either 0 or 1 indicating whether $v$ is in some $\alpha$-marking of $V$. Moreover, $\mathsf{IsMarked}(\cdot)$ takes sublinear time and when queried on all of $V$, gives a consistent $\alpha$-marking of $V$.

\begin{theorem}\label{t:ismarked}
    Let $\Phi=(V,\cC)$ be a $(k, d)$-formula. Suppose there exist constants $1/2<\beta_1<\beta_2<1-\alpha$ that satisfy the following conditions:
    \begin{align}\label{e:cond-1}
&4\alpha<2(1-\beta_2)<1-\beta_1, \notag\\
   & 16k^4d^5\leq 2^{(\beta_1-h(1-\beta_1))k}, \notag\\
    &16k^4d^5\leq 2^{(\beta_2-\beta_1)k-h\left(\frac{\beta_2-\beta_1}{1-\beta_1}\right)(1-\beta_1)k},\\
    &\delta\mapsto  (\beta_2-\delta)-h\left(\frac{\beta_2-\delta}{1-\delta}\right)(1-\delta)\text{ is decreasing on $0\leq \delta\leq \beta_1$,} \notag\\
    &2e(kd+1)< 2^{\left(1-h\left(\frac{\alpha}{1-\beta_2}\right)\right)(1-\beta_2)k}. \notag
\end{align}
(Here $h(x) := -x\log_2(x)-(1-x)\log_2(1-x)$ is the binary entropy.)

Fix constant $c > 0$. Then there exists  a $\mathrm{polylog}(n)$ time local computation algorithm $\mathsf{IsMarked}(\cdot)$ which, given any variable $v\in V$, returns an assignment $s_v\in\{0,1\}$ denoting whether $v$ is contained in an $\alpha$-marking of $\Phi$. Moreover, with probability at least $1-n^{-c}$,  the responses for all $v\in V$ yield a consistent $\alpha$-marking of $\Phi$.
\end{theorem}

The construction of $\mathsf{IsMarked}(\cdot)$ and the verification that it is a query-oblivious LCA draws inspiration from the approach in~\cite{ARV12} to obtain an oblivious LCA for hypergraph $2$-coloring. From a high level, $\mathsf{IsMarked}(\cdot)$ determines $s_v$ by randomly and greedily including variables in the marking and subsequently determining those that must/must not be in the marking, and finally (if needed) performing a rejection sampling on a smaller connected component that contains $v$. The formulae in \cref{e:cond-1} are some technical conditions that guarantee this process to go through. We defer the proof of~\cref{t:ismarked} to~\cref{app:marking}.

\section{A random local access algorithm for $k$-\textsf{SAT} solutions}

In this section, we introduce the main algorithm (\cref{alg:main}) that locally accesses the assignment of a variable $v\in V$ in a uniformly random satisfying assignment $\sigma$ to $\Phi$.

Recall from \cref{t:main} that given a  $(k,d)$-formula $\Phi=(V,\mathcal C)$, variable $v\in V$, and any constant $c>0$, we wish to come up with a random local access algorithm $\mathcal A$ such that (1) the expected running time of $\mathcal A$ is $\poly\log(n)$, and (2) the output $\widehat{\mu}_v$ of $\mathcal A$ for every $v\in V$ satisfies $d_{\TV}(\textsf{Law}(\widehat{\mu}_v), \mu_v) \le \frac{1}{n^c}$. As a high level description, given the above inputs, \cref{alg:main} performs the following:
\begin{enumerate}
    \item Locally decides whether $v$ lies in an $\alpha$-marking of $\Phi$ using $\mathsf{IsMarked}(v)$ (\cref{t:ismarked}), such that the responses for all $v\in V$ yield a consistent $\alpha$-marking $\mcM \subset V$.
    \item Suppose $\sigma\sim \mu$ is a uniformly random satisfying assignment to $\Phi$. If $v$ is marked, then we sample $\sigma(v)$ by computing  $\mathsf{MarginSample}(v)$ (adapted from~\cite[Algorithm 8]{FGWWY24}), which may make further recursive calls to sample other \textit{marked} variables.
    \item If $v$ is not marked, then we perform a depth-first search starting from $v$ to compute $\sigma$ restricted to $\mcM$. We start from $\sigma=\emptyset$; for every $w \in \mcM$ we encounter that we have not yet sampled, we compute $\mathsf{MarginSample}(w)$ and update $\sigma(w)$, to eventually obtain a (w.h.p. polylogarithmic in size) connected component $\Phi_v:=\Phi_v^\sigma$ that contains $v$. This part is carried out by the algorithm $\CONN(v)$.
    
    After obtaining the connected component $\Phi_v^\sigma$, we call $\mathsf{UniformSample}$ (\cref{t:uniform}) to sample a uniformly random satisfying assignment to $\Phi_v^\sigma$ and extend $\sigma$. We then output $\sigma(v)$.
\end{enumerate}

\begin{algorithm}
\caption{The sampling algorithm}
\label{alg:main}
 \hspace*{\algorithmicindent} \textbf{Input:} $k$-CNF formula $\Phi=(V,\cC)$ and variable $v\in V$\\
 \hspace*{\algorithmicindent} \textbf{Output:}  random assignment $\sigma(v)\in \{0,1\}$
\begin{algorithmic}[1]
\If{$\mathsf{IsMarked}(v)$}
     \State   \Return $\mathsf{MarginSample}( v)$
\Else
\State $\Phi_v \leftarrow \mathsf{Conn}( v)$ 
\State $\sigma\leftarrow \mathsf{UniformSample}(\Phi_v)$
\State \Return $\sigma (v)$ 
\EndIf
\end{algorithmic}
\end{algorithm}

To illustrate the workflow of \cref{alg:main}, we present a toy example on a small $k$-CNF formula, omitting some subroutine details for brevity.

\begin{example}Suppose $k=3$, $d=2$ and $\alpha = 1/3$.
Consider the following $(k,d)$-formula $\Phi = (V, \mathcal{C})$ with variables 
$V = \{x_1, x_2, x_3, x_4, x_5\}$
and clauses $\mathcal{C} = \{C_1, C_2, C_3\}$, where
\[
C_1 = (x_1 \vee x_2 \vee \neg x_3), \quad C_2 = (\neg x_2 \vee x_3 \vee x_4), \quad C_3 = (\neg x_4 \vee x_5 \vee \neg x_1).
\]
Suppose we wish to approximately sample $\sigma(x_1)$ using the local access algorithm $\mathcal{A}$. 
We run $\mathsf{IsMarked}$ on each variable; suppose the resulting marking is
$\mathcal{M} = \{x_2, x_5\}$,
so that $x_1$ is not marked.

Since $x_1 \notin \mathcal{M}$, we call $\mathsf{Conn}(x_1)$ to explore the connected component of $\Phi^{\mathcal{M}}$ containing $x_1$. We initialize $\mathcal{S} := \{C_1, C_3\}$ and partial assignment $\sigma = \emptyset$.

\begin{itemize}
    \item Process clause $C_1$. Since $x_2 \in \mathcal{M}$, we call $\mathsf{MarginSample}(x_2)$, which returns (say) $\sigma(x_2) = 0$. The clause becomes $(x_1 \vee 0 \vee \neg x_3) = (x_1 \vee \neg x_3)$, which is unsatisfied with the current $\sigma$. We add adjacent clause $C_2$ (via $x_2$ and $x_3$) to $\mathcal{S}$.
    
    \item Process clause $C_3$. Since $x_5 \in \mathcal{M}$, we call $\mathsf{MarginSample}(x_5)$, which returns (say) $\sigma(x_5) = 1$. The clause becomes $(\neg x_4 \vee 1 \vee \neg x_1)$, which is satisfied.
    
    \item Process clause $C_2$. We already have $\sigma(x_2) = 0$, so $\neg x_2 = 1$, and the clause is satisfied. Remove $C_2$ from $\mathcal{S}$.
\end{itemize}

Now the discovered component $\Phi_{x_1}^\sigma$ includes (1) clauses $C_1$, (2) marked variable assignments $\sigma(x_2) = 0$, $\sigma(x_5) = 1$, (3) free variables: $x_1, x_3, x_4$. We now run $\mathsf{UniformSample}$ on the reduced subformula $C_1' = (x_1 \vee \neg x_3)$. The satisfying assignments $(x_1, x_3)$ are $\{(1,0,),(1,1),(0,0)\}$.
We pick one uniformly at random, say $(x_1, x_3) = (1, 0)$. Then we return $\widehat{\mu}_{x_1} = 1$.
\end{example}

\cref{alg:main} is our main routine; it has several subroutines $\IM$, $\MS$, $\CONN$, and $\US$ that we now describe.
Recall that $\IM$ has been discussed in \cref{t:ismarked}. We now introduce the algorithm $\US$ given by the work of He--Wang--Yin \cite{HWY22}. 

\begin{theorem}[\text{\cite[Theorems 5.1 and 6.1]{HWY22}}]
\label{t:uniform}
Suppose $\Phi$ is a $(k, d)$-CNF formula on $n$ variables with $k \cdot 2^{-k} \cdot (dk) ^5\cdot 4\leq \frac{1}{150e^3}$.
    Then there exists an algorithm $\mathsf{UniformSample}$ that terminates in $O(k^3(dk)^9 n)$ time in expectation, and outputs a uniformly random satisfying assignment to $\Phi$.
\end{theorem}

As seen in \cref{t:uniform}, for an $n$-variable $(k,d)$-CNF formula, the algorithm has running time $O(n)$. However, as we will only apply $\US$  to connected components of size $\poly\log(n)$ (that are shattered by the partial assignment on the marked variables), the execution of $\US$  in Line 6 of \cref{alg:main} will only take polylogarithmic time.

We will define subroutines $\MS$ and $\CONN$ below and show that they satisfy the desired correctness and efficiency properties, beginning by verifying $\MS$ has the desired properties in~\cref{sec:ms}.

\begin{theorem}\label{t:ms}
    Let $\Phi=(V,\mathcal C)$ be a $(k,d)$-formula, $\alpha>0$, and $\mathcal M\subseteq V$ be an $\alpha$-marking as in \cref{def:marking}. Fix positive constant $c > 0$. Suppose $\theta:= 1 - \frac12 \exp \left(\frac{2 e d k }{2^{\alpha k}}\right)\geq 0.4$, $36 e d^3k^4 \cdot 0.6^{\alpha k}\leq 1/2$, and $2^{-\frac{1}{48d k^4}}\cdot e^{\frac{2d^2 /\alpha}{2^{\alpha k}}}\leq 0.9$.  Then there exists a random local access algorithm $\MS(\cdot)$ such that for every $u\in\mcM$, $\MS(u)$ returns a random value in $\{0,1\}$ that satisfies the following:
    \begin{enumerate}
    \item Let $\nu:=\mu_{\mathcal M}$ and $\widehat{\nu}$ be the joint distribution of the outputs $(\MS(u))_{u\in\mathcal M}$. Then we have $d_{\TV}(\widehat{\nu},\nu)< n^{-c}$.
        \item For every $u\in\mcM$, the expected cost of $\MS(u)$ is $\poly\log(n)$. 
    \end{enumerate}
\end{theorem}

We also require $\CONN$ to be correct and have low expected cost. 

\begin{theorem}\label{t:conn}
    Let $\Phi=(V,\mathcal C)$ be a $(k,d)$-formula, $\alpha>0$, and $\mathcal M\subseteq V$ be an $\alpha$-marking as in \cref{def:marking}. Suppose $\theta:= 1 - \frac12 \exp \left(\frac{2 e d k }{2^{\alpha k}}\right)\geq 0.4$ and $d\leq 2^{\alpha k/4}$. Then 
there exists a random local access algorithm $\CONN(\cdot )$ such that for every $v\in V\setminus \mcM$, $\CONN(v )$  returns the connected component in $\Phi^\sigma$ that contains $v$, where $\sigma$ is the partial assignment given by $(\MS(u))_{u\in\mathcal M}$. Moreover, for every $v\in V\setminus \mcM$, the expected cost of $\CONN(v)$ is $\poly\log(n)$.
\end{theorem}

From a high level, the algorithm $\CONN(v)$ explores the clauses and marked variables in the CNF formula that are reachable from $v$, greedily sampling the marked variables and expanding through unsatisfied clauses. It iteratively builds a partial assignment $\sigma$ that ``shatters'' the formula into disconnected components, isolating the one containing $v$. We will verify \cref{t:conn} in \cref{sec:conn}.

\section{Proof of \cref{t:ms}}\label{sec:ms}

In this section we show \cref{t:ms}. Throughout, let $\Phi=(V,\mathcal C)$ be a $(k,d)$-formula, $\alpha>0$, and $\mathcal M\subseteq V$ be an $\alpha$-marking with $|\mcM| = m$. 
We introduce a local access algorithm $\MS$ that satisfies the key property that the joint distribution of outputs $(\MS(u))_{u\in\mcM}$ consistently follows the marginal distribution $\mu_{\mathcal M}$. In particular,
$(\MS(u))_{u\in\mcM}$ will simulate the output of a \textit{systematic scan Glauber dynamics} on the marked variables. 

\begin{defn}\label{d:pred}
Let $\mcM = \{u_1, \ldots, u_m\}$ denote the marked variables (so $|\mcM| = m$).
\begin{itemize}
    \item For every time $t\in\ZZ$, define $i(t):=(t\mod m)+1$.
    \item For every time $t\in\ZZ$ and $u_i\in\mcM$, define $
    \pred_{u_i}(t):=
        \max\{s\leq t:i(s)=i\}.$
\end{itemize}
\end{defn}

In the systematic scan Glauber dynamics, we always resample vertex $u_{i(t)}$ at time $t$ (as opposed to randomly choosing a variable to resample at each step). For every $u\in\mcM$ and time $t$, $\pred_u(t)$ denotes the most recent time up to time $t$ at which $u$ got resampled. Observe that for all $t\in\ZZ$ and $u\in\mcM$, we have $t-m<\pred_u(t)\leq t$. Moreover, for all $w\neq u_{i(t)}$, we have $\pred_w(t)<t$.

\begin{algorithm}[H]
\caption{Systematic scan projected Glauber dynamics}\label{alg:glauber-systematic}
 \hspace*{\algorithmicindent} \textbf{Input}: a $k$-CNF formula $\Phi = (V, \mcC)$, a set of marked variables $\mcM \subset V$, time $T$, and an ordering $\mcM = \{u_1, \ldots, u_m\}$ \\
 \hspace*{\algorithmicindent} \textbf{Output}: a random assignment $X_{*} \in \{0,1\}^{\mcM}$  
\begin{algorithmic}[1]
\State Sample $X_0(u) \rand \{0, 1\}$ uniformly and independently for each $u \in \mcM$
\For{$t =1,\dots,T$}
\State Let $u := u_{i(t)}$
\State For all $u'\in \mcM\setminus\{u\}$, $X_t(u') \leftarrow X_{t-1}(u')$
\State Sample $X_t(u) \rand \mu_u(\cdot \mid {X_{t-1}(\mcM \backslash \{u\}})$
\EndFor
\State Return $X_* = X_T$
\end{algorithmic}
\end{algorithm}

We refer the readers to \cref{s:mc} for more  introduction on the systematic scan Glauber dynamics and its comparison with the (original) projected Glauber dynamics
Markov chain for sampling $k$-$\mathsf{SAT}$ solutions.

We have the following non-quantitative convergence for systematic scan Glauber dynamics.

\begin{theorem}[\cite{LPbook}]\label{t:convergence}
    Let $(X_t)_{t=0}^\infty$ denote the Glauber dynamics or the systematic scan Glauber dynamics with stationary distribution $\pi$. If $(X_t)_{t=0}^\infty$ is irreducible over  $\Omega\subseteq\{0,1\}^{\mcM}$, then for every $X_0\in \Omega$, we have $\lim_{t\to\infty}d_{\TV}(X_t,\pi)=0$.
\end{theorem}

Let $\mcM \subset V$ be a fixed set of marked variables for a given $k$-CNF $\Phi$. Let $\nu := \mu_{\mcM}$ be the marginal distribution of $\mu$ on $\mcM$. We will crucially use the following \textit{local uniformity} of~\cref{alg:glauber-systematic} (the proof in the systematic scan setting follows essentially identically to ~\cite[Lemma 5.3]{FGY21}):

\begin{lemma} \label{l:local-uniform}
Suppose $\Phi$ is a $k$-CNF with $1 < s \le \frac{2^{\alpha k}}{2 e d k}.$ Let $X \subset \{0, 1\}^{\Lambda}$ be either $X_0$ or $X_t(\mcM \backslash \{u_{i(t)}\})$ for some $t>0$ (so correspondingly, $\Lambda$ is either $\mcM$ or $\mcM \backslash \{u_{i(t)}\}$ for some $t>0$). Then for any $S \subset \Lambda$ and $\sigma: S \to \{0, 1\}$, we have
$$\P(X(S) = \sigma) \le \left( \frac12\right)^{|S|} \exp \left(\frac{|S|}{s}\right).$$
Specifically, for any $v \in \mcM$ and $c \in \{0, 1\}$, we have $\P(X(v) = c) \ge 1 - \frac12 \exp(\frac{1}{s}) \ge \frac12 - \frac{1}{s}$.
\end{lemma}

\begin{definition}\label{def:lb-dist}
Let $\pi$ be a distribution on $\{0, 1\}^{\mcM}$. We say that $\pi$  is \textit{$b$-marginally lower bounded} if for all $v \in \mcM$, $\Lambda\subseteq \mcM\setminus \{v\}$ and feasible partial assignment $\sigma_{\Lambda}\in \{0, 1,\perp\}^\Lambda$, we have $$\pi_v^{\sigma_{\Lambda}}(0), \pi_v^{\sigma_{\Lambda}}(1)\geq b.$$

Let $\pi$ be a $b$-marginally lower bounded distribution over $\{0, 1\}^{\mcM}$. For every $v\in {\mcM}$, we define the following distributions:
\begin{enumerate}
    \item \textit{Lower bound distribution }$\pi_v^{\lb}$ over $\{ 0, 1,\perp\}$: we define $\pi_v^{\lb}:=\pi^{\lb}$, with
    $$\pi^{\lb}(0)= \pi^{\lb}(1) = b, \quad  \quad\pi^{\lb}(\perp)=1-2b.$$
    \item \textit{Padding distribution} $\pi_v^{\pad, \sigma_{\Lambda}}$ over $\{0, 1\}$: for $\Lambda\subseteq \mcM\setminus\{v\}$ and feasible partial assignment $\sigma_\Lambda\in\{0, 1,\perp\}^\Lambda$, we define
    $$
    \pi_v^{\pad, \sigma_{\Lambda}}(\cdot):=\frac{\pi_v^{\sigma_{\Lambda}}(\cdot)-b}{1-2b}.
    $$
\end{enumerate}
\end{definition}

Per~\cref{l:local-uniform}, we have that $\nu = \mu_{\mcM}$ is $\theta$-lower bounded for 
\begin{equation}\label{e:theta}
\hspace{3.3cm}\theta := 1 - \frac12 \exp \left(\frac{2 e d k }{2^{\alpha k}}\right) \ge \frac12 - \frac{1}{2^{(\alpha - 2\log_2 d) k}}.
\end{equation}

\subsection{Systematic scan Glauber dynamics on marked variables}

We adapt the approach of \cite{FGWWY24} to a local sampling algorithm by simulating the systematic scan projected Glauber dynamics on $\mcM$ from time $-T$ to $0$, which is an aperiodic and irreducible Markov chain by results in \cite{FGY21,FGWWY24}.

Let $(X_t)_{-T\leq t\leq 0}$ be the output of \cref{alg:glauber-systematic}, where we relabel $X_0,\dots,X_T$ by $X_{-T},\dots,X_0$. We know from \cref{t:convergence} that, as $T\to\infty$, we have $d_{\TV}((X_0(v))_{v\in\mcM},\nu)\to 0$ where $\nu=\mu_{\mathcal M}$ is the marginal distribution of $\mu$ on the marked variables. In particular, for every fixed $n$ and $\gamma>0$, there exists $T_{\gamma}\in\NN$ such that for all $T>T_{\gamma}$, the Markov chain $(X_t)_{-T\leq t\leq 0}$ satisfies $d_{\TV}((X_0(v))_{v\in\mcM},\nu)<\gamma$.

We know from \cref{l:local-uniform} that $\nu$ is $\theta$-lower bounded, with the lower bound distribution $\nu^{\lb}$ defined by $\nu^{\lb}(0)=\nu^{\lb}(1)=\theta$ and $\nu^{\lb}(\perp)=1-2\theta$. Thus, for every $-T< t\leq 0$ and $u=u_{i(t)}$, sampling $X_t(u) \sim \nu_u^{X_{t-1}(\mcM\setminus \{u\})}$ can be decomposed into the following process: 
\begin{enumerate}
    \item With probability $2\theta$, set $X_t(u)$ to 0 or 1 uniformly at random;
    \item With probability $1-2\theta$, sample $X_t(u)$ from the padding distribution $\nu_u^{\pad, X_{t-1}( \mcM\setminus\{u\})}$.
\end{enumerate}

Our goal is to obtain $X_0(u)$ for $u \in \mcM$, which by \cref{t:convergence} will closely follow the marginal distribution $\nu_u$ for $T$ sufficiently large. It suffices to simulate the last update for $u$. The key observation here is that updates of Glauber dynamics may depend only on a very small amount of extra information. When $\theta$ is close to $1/2$, it is reasonably likely that we can determine $X_0(u)=X_{\pred_u(0)}(u)$ without any other information. Thus, we can deduce $X_0(u)$ by recursively
revealing only the necessary randomness backwards in time. This method was termed \textit{coupling towards the past} and studied for a variety of constraint satisfaction problems in~\cite{FGWWY24}.

We now give a general algorithm $\GB_{T,Y}(t,M,R)$ in \cref{alg:glauber-ms}, whose output \textit{simulates} the effects of~\cref{alg:glauber-systematic} at any particular time $t$, by looking \textit{backwards} in time at what happened over the course of running $(X_t)_{-T\leq t\leq 0}$.
The eventual algorithm $\MS(u)$ we give in \cref{t:ms} will retrieve the most recent update of each variable $u$, i.e., retrieve the coordinate $X_0(u)$.
\begin{algorithm}
\caption{$\MS(u)$}\label{alg:ms}
 \hspace*{\algorithmicindent} \textbf{Input:} a $k$-CNF formula $\Phi = (V, \mcC)$, a set of marked variables $\mcM=\{u_1,\dots,u_m\} \subset V$, and a marked variable $u\in\mcM$ \\
 \hspace*{\algorithmicindent} \textbf{Output:} a random value in $\{0,1\}$
\begin{algorithmic}[1]
\State $T\leftarrow T_{n^{-(2+c)}}$ 
\State $Y\rand \{0,1\}^{\mcM}$
\State return $\GB_{T,Y}(\pred_u(0),M=\perp^{\ZZ},R=\emptyset)$
\end{algorithmic}
\end{algorithm}

The algorithm $\GB_{T,Y}(t,M,R)$ contains another subroutine $\LB_{T,Y}(t,R)$ that is defined in \cref{alg:glauber-lb}. For every time $t$, the output of $\LB_{T,Y}(t,R)$ follows the distribution $\nu^{\lb}$ (see \cref{def:lb-dist}). In other words,  $\LB_{T,Y}(t,R)$  preliminarily decides which of the above two regimes we fell into while resampling $X_{u_{i(t)}}$ at time $t$.

Throughout~\cref{alg:glauber-ms}, we maintain two global data structures. 
\begin{itemize}
    \item We let $M:\ZZ\to\{0,1,\perp\}$ record the previously revealed outcomes of \cref{alg:glauber-ms}. That is, for every $t'\leq 0$ such that $\GB_{T,Y}(t',M,R)$ has already been executed, we set $M(t')$ to be the outcome of $\GB_{T,Y}(t',M,R)$.
    \item We let $R=\{(s,r_s)\}\subseteq \ZZ\times \{0,1,\perp\}$ record the previously revealed outcomes of \cref{alg:glauber-lb}. That is, for every $t'\leq 0$ such that $\LB_{T,Y}(t',M,R)$ has already been executed and returned $r_{t'}\in\{0,1,\perp\}$, we add $\{(t',r_{t'})\}$ to $R$. 
\end{itemize}

Since $T,Y$ remain constant throughout \cref{alg:glauber-ms,alg:glauber-lb}, and  all recursive calls access and update the same $M$ and $R$, we sometimes write $\GB(t)=\GB_{T,Y}(t,M,R)$ and $\LB(t)=\LB_{T,Y}(t,R)$ for short.

At the beginning of $\GB(t)$, we first check a few early stopping conditions:
\begin{itemize}
    \item (Lines 1--2) If variable $u_{i(t)}$ remains its initial assignment $Y(u_{i(t)})$ at the end of time $t$ (i.e., is never resampled), we terminate and return $Y(u_{i(t)})$.
    \item (Lines 3--4) If  $|R|$, the number of stored outcomes of $\LB$, already reaches $80d k^4\log n$, we terminate and return 1.
    \item (Lines 5--6) If previous iterations have already computed $\GB(t)$ and stored $M(t)\in\{0,1\}$, we terminate and return $M(t)$.
\end{itemize}

If none of the above conditions occurs, we then resample, first by applying $\LB(t)$ (\cref{alg:glauber-lb}) in Lines 7--8. If $\LB(t)\in\{0,1\}$ (which occurs with probability $2\theta$), we can update $u = u_{i(t)}$ by choosing an assignment from $\{0, 1\}$ uniformly at random without investigating the assignments of any other variables at earlier time steps (i.e., we fall into the \textit{zone of local uniformity}).

If $\LB(t)=\perp$ (which occurs with probability $1-2\theta$), then we fall into a \textit{zone of indecision} and must resample $u=u_{i(t)}$ from the padding distribution $\nu_u^{\pad, X_{t-1}(\mcM\setminus\{u\})}$. To resample its spin, we slowly proceed backwards in time, lazily computing previously resampled assignments, until we have determined enough information to compute the assignment of $u$ at the desired time step. Verifying accuracy is somewhat involved, given our lazy computation strategy and partitioning of $\nu$ into a locally uniform piece and an associated padding distribution. Thus, in~\cref{s:p-correct} we show that Lines 9--19 correctly complete the desired task, proving the following bound on $d_{\TV}((\MS(u))_{u\in\mathcal M},\nu)$.

\begin{proposition}\label{cor:tvdistance}
For any $c > 0$ and $n$ sufficiently large, we have $$d_{\TV}((\MS(u))_{u\in\mathcal M},\nu)< n^{-c}.$$
\end{proposition}

We next require that~\cref{alg:glauber-ms} has expected polylogarithmic cost. This is largely a consequence of the local uniformity of $\mu_{\mcM}$ and our lazy recursive computation of variable assignments in~\cref{alg:glauber-ms}.

\begin{lemma}\label{lem:time-gb}Suppose $2^{-\frac{1}{48d k^4}}\cdot e^{\frac{2d^2 /\alpha}{2^{\alpha k}}}\leq 0.9$.
For every $t\leq 0$, the expected cost of $\GB(t)$ is $O(k^{17}d^{10}\log^2n/\alpha)$.
\end{lemma}

\begin{algorithm}
\caption{$\LB_{T,Y}(t,R)$}\label{alg:glauber-lb}
 \hspace*{\algorithmicindent} \textbf{Input:} An integer $T\geq 0$ and assignment $Y\in\{0,1\}^{\mcM}$; an integer $t\leq 0$ \\
 \hspace*{\algorithmicindent} \textbf{Global variables:} a set $R\subseteq \ZZ\times\{0,1,\perp\}$ and  $\alpha$-marking $\mcM=\{u_1,\dots,u_m\}$ \\
 \hspace*{\algorithmicindent} \textbf{Output:} a random value in $\{0,1,\perp\}$ distributed as $\nu^{\lb}$ 
\begin{algorithmic}[1]
\If{$\pred_{u_{i(t)}}(t)\leq -T$}
\State{return $Y(u_{i(t)})$}
\EndIf
\If{$(t,r)\in R$}
\State return $r$
\EndIf
\State{Draw $x\in[0,1]$ uniformly at random}
\If{$x<2\theta$}
\State $R\leftarrow R\cup\{(t,\floor{x/\theta})\}$ and return $\floor{x/\theta}$
\EndIf
\State $R\leftarrow R\cup\{(t,\perp)\}$ and return $\perp$
\end{algorithmic}
\end{algorithm}

\begin{algorithm}[H]
\caption{$\GB_{T,Y}(t,M,R)$}\label{alg:glauber-ms}
 \hspace*{\algorithmicindent} \textbf{Input}: An integer $T\geq 0$ and assignment $Y\in\{0,1\}^{\mcM}$; an integer $t\leq 0 $\\
 \hspace*{\algorithmicindent} \textbf{Global variables}: $(k, d)$-CNF $\Phi = (V, \mcC)$, $\alpha$-marking $\mcM=\{u_1,\dots,u_m\}$, $M:\ZZ\to\{0,1,\perp\}$, and $R\subseteq \ZZ\times\{0,1,\perp\}$\\
 \hspace*{\algorithmicindent} \textbf{Output}: a random value in $\{0,1\}$  
\begin{algorithmic}[1]
\If{$\pred_{u_{i(t)}}(t)\leq -T$}
\State{return $Y(u_{i(t)})$}
\EndIf
\If{$|R|\geq 80d k^4\log n$}
\State return 1
\EndIf
\If{$M(t)\neq\perp$}
\State{return $M(t)$}
\EndIf
\If{$\LB_{T,Y}(t,R)\neq\perp$}
\State $M(t)\leftarrow \LB_{T,Y}(t,R)$ and return $M(t)$
\EndIf
\State $u\leftarrow u_{i(t)}$, $\Lambda\leftarrow \emptyset$, $\sigma_{\Lambda}\leftarrow \emptyset$, $V'\leftarrow\{u\}$ 
\While{ $\exists C\in \mcC$ such that $\vbl(C)\cap V'\neq\emptyset$, $\vbl(C)\cap (V\setminus V')\neq\emptyset$ and $C$ is not satisfied by $\sigma_{\Lambda}$}
\State choose $C$ with the lowest index
\If{for all marked $w\in \vbl(C)\setminus \{u\}$, $\LB_{T,Y}(\pred_w(t),R)$ does not satisfy $C$}
\ForAll{marked $w\in \vbl(C)\setminus \{u\}$}
\State $\Lambda\leftarrow\Lambda\cup\{w\}$ and $\sigma_{\Lambda}(w)\leftarrow \GB_{T,Y}(\pred_w(t),M,R)$
\EndFor
\State $V'\leftarrow V'\cup \text{vbl}(C)$
\Else
\ForAll{marked $w\in \vbl(C)\setminus \{u\}$}
\State $\Lambda\leftarrow \Lambda\cup\{w\}$ and $\sigma_{\Lambda}(w)\leftarrow \LB_{T,Y}(\pred_w(t),R)$
\EndFor
\EndIf
\EndWhile
\State Let $\Psi$ be the connected component in $\Phi^{\sigma_{\Lambda}}$ with $u$, and  sample $c\sim (\nu_\Psi)_u^{\pad,\sigma_{\Lambda}}$
\State $M(t)\leftarrow c$ and return $M(t)$
\end{algorithmic}
\end{algorithm}

We prove~\cref{lem:time-gb} in~\cref{s:p-efficient}. These two results together allow us to prove~\cref{t:ms}.

\begin{proof}[Proof of \cref{t:ms}]
    The theorem directly follows from combining \cref{cor:tvdistance} and \cref{lem:time-gb}. By \cref{cor:tvdistance}, we know that the joint distribution $\widehat \nu:=(\MS(u))_{u\in\mcM}$ satisfies $d_{\TV}(\widehat \nu,\nu)< n^{-c}$. By \cref{lem:time-gb}, we know that for every $u\in\mcM$, the expected cost of $\MS(u)=\GB_{T,Y}(\pred_u(0),M=\perp^{\ZZ},R=\emptyset)$ is of order $\poly\log(n)$.
\end{proof}

\section{Proof of the main theorem}

We are finally able to state and prove the formal version of \cref{t:main}. Before picking the relevant parameters, we first collect the list of conditions required to apply \cref{t:uniform,t:ms,t:conn}.

\begin{condition}
    \begin{align}\label{e:cond-2}
        &k \cdot 2^{-\alpha k} \cdot (dk) ^5\cdot 4\leq \frac{1}{150e^3},\notag \\
    &\theta := 1 - \frac12 \exp \left(\frac{2 e d k }{2^{\alpha k}}\right)\geq 0.4, \\
    & 36 e d^3k^4 \cdot 0.6^{\alpha k}\leq 1/2,\notag\\
    & 2^{-\frac{1}{48d k^4}}\cdot e^{\frac{2d^2 /\alpha}{2^{\alpha k}}}\leq 0.9, \notag\\
    &d\leq 2^{\alpha k/4}. \notag
    \end{align}
\end{condition}

Recall that we also need conditions \cref{e:cond-1} to apply \cref{t:ismarked}. We show that for $d\leq 2^{k/400}$ and $k$ sufficiently large, we can choose all the parameters appropriately so that \cref{e:cond-1,e:cond-2} are satisfied.
\begin{lemma}\label{lem:constants}
    Let 
$$
\alpha = 1/75,\qquad \beta_1 = 0.778,\qquad \beta_2 = 0.96.
$$
If $k$ is sufficiently large, and $d=d(k)\leq 2^{k/400}$, then Conditions \cref{e:cond-1,e:cond-2} are satisfied.
\end{lemma}
We defer the proof of~\cref{lem:constants} to~\cref{a:verify-consts}.
We can now state and prove  \cref{t:main-formal}, the formal version of our main result.

\begin{thm}\label{t:main-formal}
Suppose $\Phi = (V, \mcC)$ is a $(k, d)$-formula with $d\leq 2^{k/400}$ and $k$ sufficiently large. Let $\mu$ be the uniform distribution over satisfying assignments to $\Phi$, with marginal distributions $\mu_v$ for $v \in V$. Then for all $c>0$, there is a $(\poly\log (n), n^{-c})$-random local access algorithm $\mcA$ for sampling the variable assignment of $v \in V$ as $\widehat{\mu}_v$, such that 
$$d_{\TV}((\widehat{{\mu}_v})_{v\in V}, \mu) \le \frac{1}{n^c}.$$
\end{thm}
Here we remark that $c>0$ is any fixed constant, and the runtime of $\mathcal A$ depends on it. As written, both the algorithmic runtime and correctness are random, since we give expected running time and bounds on the marginal distribution in total variation distance. However, our algorithm allows derandomising either correctness or running time at the expense of worse bounds on the other.

\begin{proof}
Suppose
$\Phi = (V, \mcC)$ is a $(k, d)$-formula with $d\leq 2^{k/400}$ and $k$ sufficiently large, and $c>0$ is any constant.
Choose parameters
$$
\alpha = 1/75,\qquad \beta_1 = 0.778,\qquad \beta_2 = 0.96.
$$
By \cref{lem:constants}, we know that with these parameters, conditions \cref{e:cond-1,e:cond-2} 
are satisfied.
Thus, by \cref{t:ismarked}, there exists a $\poly\log(n)$ time oblivious local computation algorithm $\mathsf{IsMarked}(\cdot)$ that with probability at least $1 - n^{-2c}$ gives a consistent $\alpha$-marking $\mcM \subset V$.

Suppose $\mathsf{IsMarked}(\cdot)$ gives a consistent $\alpha$-marking $\mcM \subset V$. 
By \cref{t:ms}, we know that there is a random local access algorithm $\MS(\cdot)$ with expected cost $\poly\log(n)$ such that the distribution of $(\MS(u))_{u\in\mcM}\sim \widehat{\nu}$ satisfies $d_{\TV}(\widehat{\nu},\mu_{\mcM})<n^{-2c}$.

Let $\tau=(\MS(u))_{u\in\mcM}$. By \cref{prop:conn-size}, we know that for every unmarked variable $v\in V\setminus\mcM$, with probability $1-n^{-0.1\log n}$, the number of clauses in $\Phi^\tau_v$ is at most $kd\log^2 n$. We already proved in \cref{t:conn} that for every $v\in V\setminus\mcM$,  the expected cost of $\CONN(v)$ is at most $\poly\log(n)$.

Furthermore, since the reduced formula $\Phi_v^\tau$ has at least $\alpha k$ variables and at most $k$ variables in each clause, and every variable lies in at most $d$ clauses with $d \le 2^{\alpha k / 5.4}$,  by \cref{t:uniform}, the expected cost of $\mathsf{UniformSample}(\Phi^\tau_v)$ is 
asymptotically at most
$$
   k^3(dk^9) (kd\log^2n+n^{-0.1\log n}\cdot n)=\poly\log(n). 
$$
Since both $\CONN$ and $\US$ succeed with probability 1, we get that $\widehat\mu$, the joint distribution of outputs of \cref{alg:main} for all $v\in V$, satisfies $d_{\TV}(\widehat\mu,\mu)<n^{-c}$ for all $c>0$.

By construction,~\cref{alg:main} is memory-less as it samples all necessary variable assignments in order to compute the assignment of a queried variable $v$. Furthermore,~\cref{alg:main} queried on different variables $v \in V$ collectively returns an assignment $\sigma \sim \widehat \mu$ that has $d_{\TV}(\widehat \mu, \mu) < n^{-c}$. Since this holds for any $c > 0$ constant, we obtain the desired result.
\end{proof}

\section{Concluding remarks}

With more involved refinements and optimizations of the arguments in this work, the density constraint $d \lesssim 2^{k/400}$ of~\cref{t:main} can be substantially improved (to something closer to $d \lesssim 2^{k/50}$). We omit these additional calculations in favor of expositional clarity to highlight our main result: random local access models exist for arbitrary bounded degree $k$-CNFs at exponential clause density. Furthermore, these arguments can also be adapted (in a similar fashion to~\cite{GGGY21,HE23,CMM23}) to obtain a random local access model for \textit{random} $k$-CNFs in a comparable density regime.

Nonetheless, the limit of the approaches in this work would still fall well short of obtaining random local access for, e.g., approximately $d \lesssim 2^{k/4.82}$, the maximum density at which we currently know how to efficiently sample solutions to an arbitrary bounded degree $k$-CNF in nearly-linear time~\cite{HWY22,WY24}. This is because of our reliance on a query-oblivious LCA to construct a local marking and our application of weaker sampling results to a correspondingly reduced CNF.

The approach we take in this work is only one of many possible schema to reduce from existing classical, parallel, and/or distributed algorithms to more local algorithms. Our approach involved using ideas and techniques from a variety of previous works (notably~\cite{Moi19,FGWWY24,FGY21}), many of which were in the non-local setting, and adapting them in a novel way to obtain a sublinear sampler. Our approach bears some resemblance to work of~\cite{BIS17} where authors adapted a parallel Glauber dynamics algorithm to obtain random local access to proper $q$-colorings, and to work of~\cite{AJ22} that used a recursive strategy to give perfect sampling algorithms from certain spin systems in amenable graphs. We expect that many other existing algorithms (including~\cite{HWY22,HE21,HE23, WY24}) can be adapted with some work to give random local access algorithms.

\bibliographystyle{amsplain0}

\appendix

\section{Projected Glauber dynamics for sampling $k$-\textsf{SAT} solutions}\label{s:mc}

In this section, we recall some basics of Markov chains and the projected Glauber dynamics Markov chain for sampling $k$-\textsf{SAT} solutions, introduced in~\cite{FGY21}.

\begin{defn}
A \textit{Markov chain} $(X_t)_{t \ge 0}$ over a state space $\Omega$ is given by a \textit{transition matrix} $P: \Omega \times \Omega \to \RR_{\ge 0}$. A distribution $\mu$ over $\Omega$ is a \textit{stationary distribution} of $P$ if $\mu = \mu P$. 
The Markov chain $P$ is \textit{reversible} with respect to $\mu$ if it satisfies the detailed balance condition
$\mu(X) P(X, Y) = \mu(Y) P(Y, X)$, which implies that $\mu$ is a stationary distribution for $P$.

We say that a Markov chain $P$ is \textit{irreducible} if for any $X, Y \in \Omega$, there exists some integer $t$ so that $P^t(X, Y) > 0$ and \textit{aperiodic} if $\gcd\{t \mid P^t(X, X) > 0\} = 1$. If $P$ is irreducible and aperiodic, then it converges to unique stationary distribution $\mu$ with \textit{mixing time}
$$T_{\text{mix}}(P, \eps) = \max_{X_0 \in \Omega} \min \left\{t : d_{\TV}(P^t(X_0, \cdot), \mu) \le \eps \right\}.$$
\end{defn}

We now describe one specific Markov chain, the projected Glauber dynamics on the marked variables, as studied by Feng--Guo--Yin--Zhang~\cite{FGY21}. They gave the first Markov chain-based algorithm to sample from the uniform distribution of satisfying assignments to a given bounded degree $k$-CNF at exponential clause density, following earlier work of Moitra~\cite{Moi19} who gave the first fixed-parameter tractable algorithm to sample from this distribution. Since these works, there has been a tremendous amount of progress on sampling from atomic CSPs in the local lemma regime~\cite{AJ22, CMM23,FGWWY24,GGGY21,GGGH22+}.

\begin{algorithm}
\caption{Projected Glauber dynamics}\label{alg:glauber-basic}
 \hspace*{\algorithmicindent} \textbf{Input}: a $k$-CNF formula $\Phi = (V, \mcC)$, a set of marked variables $\mcM \subset V$, and time $T$ \\
 \hspace*{\algorithmicindent} \textbf{Output}: a random assignment $X_{*} \in \{0,1\}^{\mcM}$  
\begin{algorithmic}[1]
\State Sample $X_0(u) \rand \{0, 1\}$ uniformly and independently for each $u \in \mcM$
\For{$t =1,\dots,T$}
\State Sample $u \rand \mcM$
\State For all $u'\in\mcM\setminus\{u\}$, $X_t(u') \leftarrow X_{t-1}(u')$
\State Sample $X_t(u) \rand \mu_u(\cdot \mid {X_{t-1}(\mcM \backslash \{u\}}))$
\EndFor
\State Return $X_* = X_T$
\end{algorithmic}
\end{algorithm}

We contrast~\cref{alg:glauber-basic} with the \textit{systematic scan} version of the projected Glauber dynamics (\cref{alg:glauber-systematic}). Instead of resampling a random variable at time $t$, in the systematic scan Glauber dynamics, we always resample the variable $u_{i(t)}$ at time $t$, where $i(t)$ is a deterministic function of $t$.

\section{Correctness of~\cref{alg:glauber-ms}}\label{s:p-correct}

In this section, we show \cref{cor:tvdistance} that with high probability, $(\GB(\pred_u(0)))_{u\in\mcM}$ faithfully returns the final outcome $(X_0(u))_{u\in\mcM}$ of the systematic scan Glauber dynamics $(X_t)_{-T\leq t\leq 0}$ initialized at $X_{-T}=Y$. We will later use \cref{t:convergence} to show that therefore, when $T$ is set to be sufficiently large, the distribution of $(\MS(u))_{u\in\mcM}=(\GB(\pred_u(0)))_{u\in\mcM}$ is close to the marginal distribution $\nu=\mu_{\mcM}$.

\begin{proposition}\label{prop:glauber-correct}
     Fix $t\leq 0$ and let $u=u_{i(t)}$. Suppose $|R|<80d k^4\log n$ after the execution of $\GB(t)$ (i.e., Line 4 of \cref{alg:glauber-ms} is never triggered). Then $\GB(t)$ faithfully returns $X_{t}(u)$, where $(X_t)_{-T\leq t\leq 0}$ is the systematic scan Glauber dynamics started at $X_{-T}=Y$. 
\end{proposition}
\begin{proof}
   The statement clearly holds for $t=-T$ and $u=u_{i(-T)}$. Since $\pred_u(-T)= -T$, by Lines 1--2 of \cref{alg:glauber-ms}, we have $\GB(-T)=Y(u)=X_{-T}(u)$.

    Inductively, for $-T<t\leq 0$, assume the proposition for all $-T\leq t'<t$. Let $u=u_{i(t)}$. Suppose $|R|<80d k^4\log n$ after the execution of $\GB(t)$. 
     Observe that, since $\pred_w(t)< t$ for all $w\neq u$, $\GB(t)$ only makes further calls to $\GB(t')$ with $t'<t$. Thus, by the inductive hypothesis, all further calls of $\GB(t')$ have correctly returned the outcomes $X_{t'}(u_{i(t')})$.
    
    We wish to show that the resampled outcome $\GB(t)$ follows the marginal distribution $\nu_u^{X_{t-1}(\mcM\setminus\{u\})}$.  Per Lines 5--6, we may assume that $\GB(t)$ has never been called before, in which case we directly go to Line 7 of \cref{alg:glauber-ms}. Lines 7--8 guarantee that with probability $2\theta$, we assign $X_{t}(u)$ to be 0 or 1 uniformly at random.  It remains to show that in Lines 9--19, we
    are able to resample $X_{t}(u)$ from the padding distribution $\nu_u^{\pad, X_{t-1}(\mcM\setminus\{u\})}$. 
    
    To show this, we first verify that the sets $\Lambda$, $V'$ and the partial assignment $\sigma_{\Lambda}\in\{0,1,\perp\}^{\Lambda}$ obtained in Line 19 satisfy the following four conditions:
    \begin{enumerate}
        \item[(1)] $u\in V'$;
        \item[(2)] $(V'\cap \mathcal M)\subseteq \Lambda\cup\{u\}$;
        \item[(3)] for all $C\in\mathcal C$ such that $\vbl(C)\cap V'\neq\emptyset$ and $\vbl(C)\cap (V\setminus V')\neq\emptyset$, $C$ is satisfied by $\sigma_{\Lambda}$;
        \item[(4)] for all marked variables $w\in V'\setminus\{u\}$, we have $\sigma_{\Lambda}(w)=X_{t-1}(w)\in\{0,1\}$; for all variables $w\in\Lambda\setminus V'$, we either have $\sigma_{\Lambda}(w)=X_{t-1}(w)\in\{0,1\}$, or have  $\sigma_{\Lambda}(w)=\perp$.
    \end{enumerate}
    Here, property (1) holds because $u$ is added to $V'$ in the initialization, and $V'$ never removes variables. Property (2) holds because if variables in some clause $C$ are added to $V'$ in Line 15, then all marked variables in $C$ are added to $\Lambda$ in Line 14. As the while loop terminates, the opposite condition of Line 10 holds, which is exactly property (3). 
    
    We now show property (4). For every $w\in V'\setminus\{u\}$, we know that $w$ is added to $V'$ in Line 15 due to some clause $C$; if $w$ is marked, then by Line 14 and the inductive hypothesis, we know that we have assigned $\sigma_{\Lambda}(w)\leftarrow\GB(\pred_w(t))=X_{t-1}(w)\in\{0,1\}$. For every $w\in \Lambda\setminus V'$,  by Line 18, we have assigned $\sigma_{\Lambda}(w)\leftarrow \LB(\pred_w(t))$. If $\LB(\pred_w(t))\neq\perp$, then we have $\LB(\pred_w(t))=X_{\pred_w(t)}(w)=X_{t-1}(w)\in\{0,1\}$.

    Let $\Psi$ denote the connected component in $\Phi^{\sigma_{\Lambda}}$ that contains $u$. 
    Let $\mu_\Psi$ denote the distribution of a uniformly random satisfying assignment to $\Psi$. By property (3), we know that the connected component in $\Phi^{\sigma_{\Lambda}}$ that contains $u$ is supported on $V'$. By property (4), we know that $X_{t-1}(\mcM\setminus\{u\})$ is an extension of $\sigma_\Lambda$, which means that the connected component in $\Phi^{X_{t-1}(\mcM\setminus\{u\})}$ that contains $u$ is also supported on $V'$. Moreover, property (4)  implies that $\sigma_{\Lambda}(V'\setminus\{u\})=X_{t-1}(V'\setminus\{u\})$, which means that the two marginal distributions $(\mu_{\Psi})_u^{\sigma_{\Lambda}}$ and $(\mu_{\Psi})_u^{X_{t-1}(V'\setminus\{u\})}$ are the same. Altogether, we get that
    $$
    \mu_u^{\sigma_{\Lambda} }=(\mu_\Psi)_u^{\sigma_{\Lambda}}=(\mu_\Psi)_u^{X_{t-1}(\mathcal M\setminus \{u\})}=\mu_u^{X_{t-1}(\mathcal M\setminus \{u\})}.
    $$
    Recall that $\nu $ is the marginal distribution of $\mu$ on $\mcM$. Let $\nu_\Psi$ denote the marginal distribution of $\mu_\Psi$ on $\mcM$. Since $X_{t-1}(\mathcal M\setminus \{u\})$ and $\sigma_{\Lambda}$ are both supported on subsets of $\mcM$, the above gives
    $$\nu_u^{\sigma_{\Lambda} }=(\nu_\Psi)_u^{\sigma_{\Lambda}}=(\nu_\Psi)_u^{X_{t-1}(\mathcal M\setminus \{u\})}=\nu_u^{X_{t-1}(\mathcal M\setminus \{u\})}.$$

    Recall that we wish to sample from $\nu_u^{\pad, X_{t-1}(\mcM\setminus\{u\})}$. Observe that for any partial assignment $\sigma$, $\nu_u^{\pad,\sigma}$ is a deterministic function of $\nu_u^{\sigma}$ (see \cref{def:lb-dist}). Since
$\nu_u^{X_{t-1}(\mcM\setminus\{u\})}=(\nu_\Psi)_u^{\sigma_{\Lambda}}$, we have $\nu_v^{\pad,X_{t-1}(\mcM\setminus\{u\})}=(\nu_\Psi)_v^{\pad,\sigma_{\Lambda}}$ as well. Thus, it suffices to sample $c\sim (\nu_\Psi)_u^{\pad,\sigma_{\Lambda}}=\nu_u^{\pad,X_{t-1}(\mcM\setminus\{u\})}$ which was performed in Line 19. This shows that Lines 9--19 draws $X_t(u)$ from the padding distribution $\nu_u^{\pad, X_{t-1}(\mcM\setminus\{u\})}$, and finishes the proof.
\end{proof}

We now show that the condition $|R|<80d k^4\log n$ happens with high probability. To do this, we show that in a single instance of $\GB(t)$, the ``chain'' of further recursions $\GB(t')$  is unlikely to be large. We build the following graph $G$ to track these recursions.

\begin{definition}
    For every $C\in \mathcal C$ and $t\leq 0$, let 
$$\phi(C,t):=(\vbl(C), \{\pred_w(t):w\in 
\vbl(C)\cap \mcM \}),$$
i.e., $\phi(C, t)$ is the pair comprising the variables of $C$ and the most recent times that any marked variable in $C$ was resampled up until time $t$.
Consider an associated graph $G_t$ defined by
$$
V(G_t)=\{\phi(C,t'):C\in \mathcal C,\, -T\leq t'\leq t\},
$$
such that $\phi(C_1,t_1)\sim \phi(C_2,t_2)$ in $G_t$ if and only if the following holds:
\begin{enumerate}
    \item $\vbl(C_1)\cap \vbl(C_2)\neq\emptyset$,
    \item For $\mathcal T:=\{\pred_w(t_1):w\in 
\vbl(C_1)\cap \mcM \}\cup \{\pred_w(t_2):w\in 
\vbl(C_2)\cap \mcM \}$, we have $\max\mathcal T-\min\mathcal T<2m$ (recall that $m=|\mcM|$).
\end{enumerate}
\end{definition}

We also recall the notion of a 2-tree in a graph or hypergraph.
\begin{definition}\label{def:2tree}
    Let $G=(V(G),E(G))$ be a graph or hypergraph. We say that $Z\subseteq V(G)$ is a 2-tree if $Z$ satisfies the following conditions:
\begin{enumerate}
    \item for all $u,v\in Z$, $\text{dist}_{G}(u,v)\geq 2$;
    \item the associated graph with vertex set $Z$ and edge set $\{\{u,v\}\subseteq Z: \dist_{G}(u,v)=2\}$ is connected.
\end{enumerate}
\end{definition}

There are not many $2$-trees containing a fixed vertex in a graph of bounded maximum degree. We recall one example upper bound below that will be sufficient for our purposes.
\begin{observation}[\text{see \cite[Corollary 5.7]{FGY21}}]\label{obs:2tree}
    Let $G=(V(G),E(G))$ be a graph or hypergraph with maximum degree $D$. Then for every $v\in V$, the number of 2-trees  $Z\subseteq V(G)$ containing $v$ with $|Z|=\ell$ is at most $\frac{(eD^2)^{\ell-1}}{2}$.
\end{observation}

The following proposition shows that the size of $R$ is unlikely to be large when we terminate $\GB(t)$ for any $t\leq 0$.

\begin{proposition}\label{prop:glauber-time}
Fix $t\leq 0$. Suppose $\theta:= 1 - \frac12 \exp \left(\frac{2 e d k }{2^{\alpha k}}\right)\geq 0.4$ and $36 e d^3k^4 \cdot 0.6^{\alpha k}\leq 1/2$.
Upon the termination of $\GB(t)$, for every $\eta\geq 1$, the size of $R$ satisfies
$$
\PP[|R|\geq 24d k^4(\eta+1)]\leq 2^{-\eta}.
$$
\end{proposition}
\begin{proof}
Fix $t\leq 0$ and consider some instance of $\GB(t)$. 
For $t_1<t_0\leq t$, we say that $\GB(t_1)$ is \textit{triggered} by $x:=\TS(C,t_0)$ if $\GB(t_1)$ is called in Line 14 of $\GB(t_0)$ with clause $C$.
Let $W=\{x\in V(G_t):x\text{ triggers recursive calls}\}$. We begin by verifying a few basic properties that $R$ and $G := G_t$ enjoy.

\begin{claim}\label{claim:rbound-1}
    Upon the termination of $\GB(t)$, we have $|R|\leq 2k^2|W|+2k$.
\end{claim}
\begin{proof}

Observe that for every $t_0\leq t$, $\GB(t_0)$ calls $\LB$ at most $k+1$ times (in Line 7 and Line 12 of \cref{alg:glauber-ms}) before possibly going into another subroutine $\GB(t_1)$. Let $A$ denote the set of timestamps $t_0$ such that $\GB(t_0)$ was called at least once. Then we have $|R|\leq (k+1)|A|$. 

Observe that every $\GB(t_0)$ with $t_0<t$ is triggered by some $x\in W$. Moreover,  every $x\in W$ triggers at most $k$ subroutines $\GB(t_0)$ in Line 14. Thus, we get that
$|A|\leq k|W|+1$,
which gives $|R|\leq (k+1)|A|\leq (k^2+k)|W|+k+1\leq 2k^2|W|+2k$.
\end{proof}

\begin{claim}\label{claim:rbound-2}
    The maximum degree of $G$ is at most $6 k^2d-1$.
\end{claim}
\begin{proof}
     Fix $\phi(C_1,t_1)\in W$. There are at most $kd$ clauses $C_2\in\mcC$ such that $\vbl(C_1)\cap \vbl(C_2)\neq\emptyset$. 
     For any such $C_2$, we count the number of possible $\TS(C_2,t_2)$ so that $\mathcal T=\{\pred_w(t_1):w\in 
\vbl(C_1)\cap \mcM \}\cup \{\pred_w(t_2):w\in 
\vbl(C_2)\cap \mcM \}$ satisfies $\max\mathcal T-\min\mathcal T<2m$.

     Suppose
     \begin{align*}
         \{\pred_w(t_1):w\in 
\vbl(C_1)\cap \mcM \}&=\{s_1,\dots,s_{k_1}\}\text{ with }s_1<\dots<s_{k_1},\\
         \{\pred_w(t_2):w\in 
\vbl(C_2)\cap \mcM \}&=\{s_1',\dots,s_{k_2}'\}\text{ with }s_1'<\dots<s_{k_2}'.
     \end{align*}
     Observe that $s_{k_1}\leq t_1<s_1+m$ and $s'_{k_2}\leq t_2<s_1'+m$. If $\max\mathcal T-\min\mathcal T<2m$, then we have 
     $$
     s_{k_1}-2m<s_1'<s_{k_2}'<s_1+2m,
     $$
     which gives $t_2<s_1'+m<s_1+3m$ and $t_2\geq s_{k_2}'>s_{k_1}-2m$. Thus, we have $$s_{k_1}-2m<t_2<s_1+3m.$$
     Let $S:=\{s_{k_1}-2m+1,s_{k_1}-2m+2,\dots, s_1+3m-1\}$. In particular, $S$ is an interval of size $\leq 5m$ given by $\phi(C_1,t_1)$.

     Observe that as $t_2$ increments from $s_{k_1}-2m+1$ to $s_1+3m-1$, we have $\{\pred_w(t_2):w\in 
\vbl(C_2)\cap \mcM \}\neq \{\pred_w(t_2-1):w\in 
\vbl(C_2)\cap \mcM \}$ only if $\pred_{w}(t_2)> \pred_{w}(t_2-1)$ for some $w\in\vbl(C_2)\cap\mcM$. Moreover, since  $|\vbl(C_2)\cap\mcM|\leq k$ and $|S|\leq 5m$,
we know that there are at most $5k$ such numbers $t_2$ in $S$, and these numbers have been completely determined by $\vbl(C_2)$ and $S$ (i.e., by $\phi(C_1,t_1)$ and $C_2$). These numbers partition $S$ into at most $5k+1$ intervals such that $\{\pred_w(t_2):w\in 
\vbl(C_2)\cap \mcM \}$ is the same for all $t_2$ on each interval. Thus, for every fixed $\phi(C_1,t_1)$ and $C_2$, the set $\{\{\pred_w(t_2):w\in 
\vbl(C_2)\cap \mcM \}: t_2\in S\}$ has size at most $5k+1$. 

     Therefore, given any $\phi(C_1,t_1)$, we can pick a neighbor $\phi(C_2,t_2)\sim \phi(C_1,t_1)$ by first picking $C_2$ (which has $\leq kd$ choices) and then picking an element in $\{\{\pred_w(t_2):w\in 
\vbl(C_2)\cap \mcM \}: t_2\in S\}$  (which has $\leq 5k+1$ choices). So the number of possible $\phi(C_2,t_2)\sim \phi(C_1,t_1)$ in $G$ is at most $kd(5k+1)\leq 6k^2d-1$.
\end{proof}

\begin{claim}\label{claim:rbound-2.5}
    Let $u=u_{i(t)}$ and $W'=\{\phi(C,t)\in W:u\in \vbl(C)\}$. Then $G[W']$ is a clique.
\end{claim}
\begin{proof}
    Consider any two clauses $C,C'$ such that $u\in\vbl(C)\cap \vbl(C')$. Suppose $\phi(C,t),\phi(C',t)\in W'$. Clearly $\vbl(C)\cap \vbl(C')\neq\emptyset$. Moreover, since the timestamps $\pred_w(t)$ over all marked variables $w\in\mcM$ lie in the range $\{t-m+1,\dots,t\}$, we get that the maximum and minimum of $\{\pred_w(t):w\in 
\vbl(C)\cap \mcM \}\cup \{\pred_w(t):w\in 
\vbl(C')\cap \mcM \}$  differ by at most $m-1<2m$. Thus we have $\phi(C,t)\sim \phi(C',t)$ in $G$.
\end{proof}

\begin{claim}\label{claim:rbound-3}
Let $W'$ be as in \cref{claim:rbound-2.5}.
For every $x\in W$, there exists a path $p_0\dots p_\ell$ in $G[W]$ such that $p_0\in W'$ and $p_\ell=x$.
\end{claim}
\begin{proof}
    Let $x=\phi(C,t_1)$ be any element in $W$. We perform a double induction, the outside on $t_1$ and the inside on $C$.

    \begin{itemize}
        \item \textbf{Base case: $t_1=t$.} 
        
        Suppose first that $t_1=t$, so $x=\TS(C,t)$ for some $C$. Let $u=u_{i(t)}$. If $u\in \vbl(C)$, then $x\in W'$ and we are done. Inductively, suppose \cref{claim:rbound-3} holds for all $x'=\TS(C',t)$ that triggers a recursive call earlier than $x$ in \cref{alg:glauber-ms}. If $u\notin \vbl(C)$, then by the while loop condition Line 10, there must exist some $x'=\TS(C',t)\in W$ such that $x'$ triggers a recursive call earlier than $x$, and $\vbl(C)\cap \vbl(C')\neq\emptyset$. By the inductive hypothesis,  there exists  a path $p_0\dots p_\ell$ in $G[W]$ such that $p_0\in W'$ and $p_\ell=x'$. Since the maximum and minimum of $\{\pred_w(t):w\in 
\vbl(C)\cap \mcM \}\cup \{\pred_w(t):w\in 
\vbl(C')\cap \mcM \}$  differ by at most $m-1<2m$, we get that $x'\sim x$ in $G$. Therefore we can extend the path $p_0\dots p_\ell$ with $p_{\ell +1}=x$. This finishes the inductive step.

    \item \textbf{Inductive step: $t_1<t$.}
    
    Now suppose $x=\TS(C,t_1)$ with $t_1<t$. By the inductive hypothesis, we can assume \cref{claim:rbound-3} for all $t_0\in\{t_1+1,\dots,t\}$. Let $u_1=u_{i(t_1)}$.
    
    Suppose first that $u_1\in \vbl(C)$. Since $t_1<t$, there must exist $t_0\in\{t_1+1,\dots,t\}$ and $C'\in\mathcal C$ such that $\TS(C',t_0)$ triggers $\GB(t_1)$, with $u_1\in \vbl(C')$ and $t_1=\pred_{u_1}(t_0)$. Let $y=\TS(C',t_0)$.
    Clearly $\vbl(C)\cap \vbl(C')\neq\emptyset$. Since $t_0-m< t_1< t_0$, we also know that the maximum and minimum of $\{\pred_w(t_1):w\in 
\vbl(C)\cap \mcM \}\cup \{\pred_w(t_0):w\in 
\vbl(C')\cap \mcM \}$ differ by at most $2m-1<2m$. Thus we have $x\sim y$ in $G$. By the inductive hypothesis for $t_0$, there exists  a path $p_0\dots p_\ell$ in $G[W]$ such that $p_0\in W'$ and $p_\ell=y$. Since $x\sim y$ in $G$, we can extend the path by $p_{\ell+1}=x$.

    Inductively, suppose $u_1\notin \vbl(C_1)$, and \cref{claim:rbound-3} holds for all $t_0\in\{t_1+1,\dots,t\}$ and for all $x'=(C',t_1)$ that triggers a recursive call earlier than $x$. Then there must exist some $x'=\TS(C',t_1)\in W$ such that $x'$ triggers a recursive call earlier than $x$, and $\vbl(C)\cap \vbl(C')\neq\emptyset$.
    By the inductive hypothesis,  there exists  a path $p_0\dots p_\ell\in G[W]$ such that $p_0\in W'$ and $p_\ell=x'$. Since the maximum and minimum of $\{\pred_w(t_1):w\in 
\vbl(C)\cap \mcM \}\cup \{\pred_w(t_1):w\in 
\vbl(C')\cap \mcM \}$  differ by at most $m-1<2m$, we get that $x'\sim x$ in $G$. Thus we can extend the path by $p_{\ell+1}=x$. This finishes the inductive step.
    \end{itemize}
\end{proof}

\begin{claim}\label{claim:rbound-4}
    For all $x=\TS(C,t')\in W$ and $w\in\vbl(C)\cap\mcM$, the function $\mathsf{LB}\text{-}\mathsf{Sample}(\pred_w(t'))$ does not satisfy $C$.
\end{claim}
\begin{proof}
    This directly follows from Lines 12--14 of \cref{alg:glauber-ms}. Since $\TS(C,t')$ triggers a recursion in Line 14, by the condition in Line 12, for all marked variables $w\in \vbl(C) \cap \mcM$, the function $\LB(\pred_w(t'))$ does not satisfy $C$.
\end{proof}

With these claims, we can now prove the proposition.
Fix $t\leq 0$ and $\eta\geq 1$. Assume $|R|\geq 24d k^4(\eta+1)$, which by \cref{claim:rbound-1} implies that $|W|\geq 6d k^2 (\eta+1)$. By \cref{claim:rbound-3}, we know that  $W\cap W'\neq\emptyset$; by \cref{claim:rbound-2}, $G$ has maximum degree $\leq 6d k^2-1$.  Thus, by a greedy selection, we can find a 2-tree $Z\subseteq W$ containing some element in $W'$ such that $|Z|=\eta+1$. 
    
    Fix any  2-tree $Z\subseteq V(G)$ such that $|Z|=\eta+1$. For every $x=\TS(C,t')\in Z$, if $x\in W$, then we know from \cref{claim:rbound-4} that $\LB(\pred_w(t'))$ does not satisfy $C$ for all $w\in\vbl(C)\cap\mcM$. Since the latter happens with probability at most $(1-\theta)^{\alpha k}$, we have
    $$
    \PP[x\in W]\leq (1-\theta)^{\alpha k}.
    $$
    Since $Z$ is a 2-tree, we know that for every two $x=\phi(C_1,t_1),y=\phi(C_2,t_2)\in Z$, we   have $\{\pred_w(t_1):w\in\vbl(C_1)\cap \mcM\}\cap \{\pred_w(t_2):w\in\vbl(C_2)\cap \mcM\}=\emptyset$ (as otherwise $C_1$ and $C_2$ would share a variable, and the union of these two sets would span an interval of size at most $2(m-1)+1=2m-1$, meaning that $x\sim y$ in $G$, which contradicts the fact that $Z$ is an independent set in $G$). In particular, the sets
    $$
     \{\{\pred_w(t'):w\in\vbl(C)\cap \mcM\}: \TS(C,t')\in Z\}
    $$
    are disjoint.
    Thus, for every fixed 2-tree $|Z|=\eta+1$ in $V(G)$, we have
    $$
    \PP[Z\subseteq W]\leq \PP[x\in W\text{ for all $x\in Z$}]\leq (1-\theta)^{\alpha k\eta}.
    $$
    
    Let $\mathcal T$ denote the set of 2-trees of $V(G)$ of size $\eta+1$ that intersects with $W'$. Since $|W'|\leq d$ and $G$ has maximum degree at most $6k^2d-1$, by \cref{obs:2tree}, we have
    $$
    |\mathcal T|\leq  d\cdot \frac{(e(6k^2d)^2)^{\eta}}{2}\leq (36 e d^3k^4)^\eta.
    $$
    Thus, we have
    $$
    \PP[|W|\geq 6d k^2 (\eta+1)]\leq\sum_{Z\in\mathcal T}\PP[Z\subseteq W]\leq  (36 e d^3k^4)^{\eta}(1-\theta)^{\alpha k\eta}\leq 2^{-\eta},
    $$
    where the last step used the assumption that $\theta\geq 0.4$ and $36 e d^3k^4 \cdot 0.6^{\alpha k}\leq 1/2$.
    This implies that 
    $$
    \PP[|R|\geq 24d k^4(\eta+1)]\leq \PP[|W|\geq 6d k^2(\eta+1)]\leq 2^{-\eta}.
    $$
    
\end{proof}
Setting $\eta=(3+c)\log n$, we get the following correctness statement on \cref{alg:glauber-ms}.
\begin{corollary}
    For every $t\leq 0$, we have $$\PP[\GB(t)\neq X_t(u)]\overset{\ref{prop:glauber-correct}}{\leq} \PP[|R|>80d k^4\log n]\leq  \PP[|R|>24d k^4((3+c)\log n+1)] \overset{\ref{prop:glauber-time}}{\leq} n^{-(3+c)}.$$
    In particular, for every $u\in\mcM$, we have $$\PP[\GB(\pred_u(0))\neq X_0(u)]=\PP[\GB(\pred_u(0))\neq X_{\pred_u(0)}(u)]\leq n^{-(3+c)}.$$
    Taking a union bound over all $u\in\mathcal M$, we get that
    $$
    \PP[(\MS(u))_{u\in\mathcal M}\neq X_0]\leq n^{-(2+c)}.
    $$
\end{corollary}

Since we have picked $T$ in \cref{alg:ms} sufficiently large so that $d_{\TV}(X_0,\nu)\leq n^{-(2+c)}$, we get that the joint output $(\MS(u))_{u\in\mathcal M}$ satisfies $d_{\TV}((\MS(u))_{u\in\mathcal M},\nu)\leq n^{-c}$, proving~\cref{cor:tvdistance}.

\section{Efficiency of~\cref{alg:glauber-ms}}\label{s:p-efficient}

We now move on to show the efficiency of $\GB(t)$ for all $t\leq 0$.
Observe that for every $r\geq 48d k^4$, by \cref{prop:glauber-time}, we have
\begin{align}\label{eq:rsize}
     \hspace{4.8cm}\PP[|R|\geq r]\leq 2^{-\frac{r}{48d k^4}}.
\end{align}
Moreover, we terminate $\GB(t)$  once we reach $|R|\geq 80d k^4\log n$.
We will use these information to give an upper bound on the expected cost of $\GB(t)$.

We start by upper bounding the size of the final sets $V'$ and $\Lambda$ in  \cref{alg:glauber-ms} in terms of $|R|$. 

\begin{lemma}
    The $V'$ and $\Lambda$  in Line 19 of \cref{alg:glauber-ms} satisfy
    $$
    |\Lambda|\leq kd|V'|\leq 2kd^2|R|/\alpha.
    $$
\end{lemma}
\begin{proof}
    By Line 10 of \cref{alg:glauber-ms}, we know that for all $u\in\Lambda$, there exists clause $C\in \mcC$ such that $u\in \vbl(C)$ and $\vbl(C)\cap V'\neq\emptyset$. This shows that $|\Lambda|\leq kd|V'|$. 
    Observe that $V'$ contains at least $|V'|/(\Delta+1)\geq |V'|/(2dk)$ clauses with disjoint variable sets.
    Since every marked variable in the clauses added to $V'$ have gone through at least one round of $\LB$ (per Line 5 and Line 12), we get that $|R|\geq \alpha k\cdot |V'|/(2dk)=\alpha|V'|/(2d)$. 
\end{proof}

To upper bound the expected cost of Line 19, we further need the result of He--Wang--Yin \cite{HWY22} on the existence of a ``Bernoulli factory'' algorithm $\mathsf{BF}(\cdot)$ such that for every locally uniform CNF and variable $v$, the Bernoulli factory efficiently draws a $\{0,1\}$-random variable according to the padding distribution of $v$.

\begin{proposition}[{\cite[Lemma 3.10 and Appendix A]{HWY22}}]\label{prop:bf}
    Let $\Psi=(V_{\Psi},\mathcal C_{\Psi})$ be a $k$-CNF, $\sigma$ be a feasible partial assignment of $\Psi$, and $\Psi^\sigma=(V_{\Psi}^\sigma,\mcC_{\Psi}^\sigma)$ be the reduced CNF (see \cref{def:reduced}). Suppose we have
    $$
    \PP[\lnot C\mid\sigma]\leq \zeta \qquad \text{for every $C\in\mcC_{\Psi}^\sigma$}.
    $$
    Then there exists an algorithm $\mathsf{BF}(\cdot)$ such that for every $v\in V^\sigma$, 
    \begin{itemize}
        \item  $\mathsf{BF}(v)$ with probability 1  returns $\xi\sim (\mu_\Psi)_v^{\mathrm{pad},\sigma}$;
        \item $\mathsf{BF}(v)$ has expected cost $O(k^9d^6|\mcC_{\Psi}^\sigma|(1-e\zeta)^{-|\mcC_{\Psi}^\sigma|})$.
    \end{itemize}
\end{proposition}

We remark that in \cref{alg:glauber-ms}, our partial assignment $\sigma_{\Lambda}$  is always supported on the marked variables. Since every clause has at least $\alpha k$ unmarked variables that are not assigned by $\sigma_{\Lambda}$, we get  that $\PP[\lnot C \mid \sigma_{\Lambda}]\leq 2^{-\alpha k}$ for every $C$.
Thus when applying \cref{prop:bf}, we will take $\zeta=2^{-\alpha k}$.

We can now give an upper bound on the expected cost of \cref{alg:glauber-ms}, proving~\cref{lem:time-gb}.
\begin{proof}[Proof of~\cref{lem:time-gb}]
Let $A$ be the set of all $t'$ such that $\GB(t')$ is executed at least once as subroutine of $\GB(t)$. Since we run $\LB(t')$ at the beginning (Line 7) of each round $\GB(t')$, we know that $|A|\leq |R|$.

Observe that for every $t'$, if $\GB(t')$ has been computed once, then we will permanently assign $M(t')=\GB(t')\in\{0,1\}$, and later executions of $\GB(t')$ will terminate at Line 4. Thus, it suffices to upper bound the cost of every first execution of $\GB(t')$ before entering another $\GB(t'')$. Multiplying this by $|A|$ would give an upper bound on the cost of $\GB(t)$.

Suppose we are at the first time of executing $\GB(t')$ for some $t'$.
We first estimate the cost of the while loop in Lines 10--18, which is a constant multiple of the number of executions of Line 14 and Line 18. Note that every time we execute Line 14 or Line 18, some $w$ is added to $\Lambda$ due to a clause $C$ chosen in Line 11, with $w\in\vbl(C)$.
Moreover, each clause $C$ can be chosen in Line 11 at most once. Thus, the number of pairs $(w,C)$ that correspond to an execution of Line 14 or Line 18 is at most $d|\Lambda|$. Consequently, the cost of the while loop in Lines 10--18 is $O(d|\Lambda|)=O(kd^3|R|/\alpha)$.

We now estimate the cost of Line 19. Observe that in Line 19, we can apply the Bernoulli factory in \cref{prop:bf} to compute the padding distribution $(\nu_{\Psi})_u^{\pad,\sigma_{\Lambda}}$.
Since the connected component $\Psi$ in Line 19 has at most $d|V'|\leq 2d^2|R|/\alpha$ clauses,   by \cref{prop:bf}, the expected cost of Line 19 is at most
$$
O(k^9d^8|R|/\alpha(1-e2^{-\alpha k})^{-2d^2|R|/\alpha}).
$$
Let $R_{\max}=80dk^4\log n$.
Combining the above and applying \cref{eq:rsize}, we get that the expected cost of $\GB(t)$ is at most
\begin{align*}
&\EE[|A|\cdot O(kd^3|R|/\alpha+k^9d^8|R|/\alpha(1-e2^{-\alpha k})^{-2d^2|R|/\alpha})]\\
&\leq \EE[|R|\cdot O(kd^3|R|/\alpha+k^9d^8|R|/\alpha(1-e2^{-\alpha k})^{-2d^2|R|/\alpha})]\\
    &\leq \sum_{r=0}^{R_{\max}} \PP[|R|\geq r]\cdot r\cdot O(kd^3r/\alpha+k^9d^8r/\alpha(1-e2^{-\alpha k})^{-2d^2 r/\alpha})\\
    &\leq \sum_{r=0}^{R_{\max}} O(2^{-\frac{r}{48dk^4}}\cdot r\cdot kd^3r/\alpha+k^9d^8r/\alpha(1-e2^{-\alpha k})^{-2d^2 r/\alpha}) \\   &\leq  \sum_{r=0}^{R_{\max}} O\left(2^{-\frac{r}{48dk^4}}\cdot r\cdot\left(\frac{kd^3r}{\alpha}+\frac{k^9d^8r}{\alpha} \cdot e^{\frac{2d^2 r/\alpha}{2^{\alpha k}}}\right)\right)=O(k^9d^8R_{\max}^2/\alpha)=O(k^{17}d^{10}\log^2n/\alpha),
\end{align*}
where in the last step we used the information that
$$
2^{-\frac{1}{48d k^4}}\cdot e^{\frac{2d^2 /\alpha}{2^{(1-\alpha)k}}}\leq 0.9 \quad \Longrightarrow\quad  \sum_{r=0}^{R_{\max}}\left(2^{-\frac{1}{48d k^4}}\cdot e^{\frac{2d^2 /\alpha}{2^{\alpha k}}}\right)^r=O(1).
$$
\end{proof}

\section{Proof of~\cref{t:conn}}\label{sec:conn}

In this section we show \cref{t:conn}. Let $\Phi=(V,\mathcal C)$ be a $(k,d)$-formula, $\alpha>0$, and $\mathcal M\subseteq V$ be an $\alpha$-marking. 
We introduce the local access algorithm $\CONN$ (\cref{alg:conn}) such that for every unmarked variable $v\in V \setminus \mcM$, $\CONN(v)$ returns the connected component of $\Phi^{(\MS(u))_{u\in\mcM}}$ that contains $v$. 

\begin{algorithm}[ht!]
\caption{$\mathsf{Conn}(v)$}
\label{alg:conn}
 \hspace*{\algorithmicindent} \textbf{Input:} CNF formula $\Phi=(V,\mathcal C)$, $v\in V\setminus\mcM$\\
 \hspace*{\algorithmicindent} \textbf{Output:}  random assignment in $\{0,1\}$ 
\begin{algorithmic}[1]
\State Initialize $\mcS = \{C \in \mcC \mid v\in \text{vbl}(C)\}$ as a stack
\While{$|\mcS| > 0$}
\State Pop $C$ off of $\mcS$
\For{$w \in \vbl(C^{\sigma})$}
\If{$\mathsf{IsMarked}(w)$}
\State $\sigma(w) \leftarrow \mathsf{MarginSample}(w)$
\State Remove from $\mcS$ any clauses satisfied by $\sigma$ 
\EndIf
\If{$C^{\sigma} = 1$}
\State \textbf{break} 
\EndIf
\EndFor
\If{$C^{\sigma} \neq 1$}
\For{$w \in \vbl(C^{\sigma})$}
\State Push $\{ C'\in\mcC \mid  w \in \vbl((C')^{\sigma}), (C')^{\sigma} \neq 1, C'\text{ has not been in $\mcS$}\}$ onto $\mcS$
\EndFor
\EndIf
\EndWhile
\Return $\Phi_v^{\sigma}$
\end{algorithmic}
\end{algorithm}

For every clause $C\in\mcC$ and partial assignment $\sigma$, we define $C^\sigma$ to be the reduced clause after conditioning on $\sigma$ and removing the variables assigned 0 or 1 by $\sigma$ (in particular, $C^\sigma=1$ if $C$ is already satisfied by $\sigma$). It is not hard to read from the algorithm that $\CONN(v)$ computes the connected component of $\Phi^{(\MS(u))_{u\in\mcM}}$ that contains $v$. Starting with the empty assignment $\sigma=\emptyset$ and $\mcS$ to be the clauses that contain $v$, the algorithm keeps sampling the marked variables touched by $\mathcal S$, adding these assignments to $\sigma$, and see if clauses in $\mathcal S$ are satisfied by $\sigma$. If a clause is satisfied by $\sigma$, then we simply remove it from $\mathcal S$; if it is not satisfied by $\sigma$, then in addition to removing it from $\mathcal S$, we also add all unexplored clauses adjacent to this clause to $\mathcal S$. One can see that in this procedure,  $\mcS$ maintains to be the boundary of those explored and yet unsatisfied clauses given by the partial assignment $\sigma$. Once $\mathcal S$ becomes empty, it means that in the reduced CNF $\Phi^\sigma$, there is no unexplored clause adjacent to the  explored clauses that are not satisfied by $\sigma$. In other words, $\sigma$ has shattered $\Phi$ so that those explored clauses form a connected component $\Phi_v^\sigma$.

To show \cref{t:conn}, it remains to show that the algorithm $\CONN(v)$ has expected cost $\poly\log(n)$. We show this by giving an upper bound on the size of the connected component of $\Phi^{(\MS(v))_{u\in\mcM}}$ that contains $v$.

\begin{proposition}\label{prop:conn-size}
 Fix $v\in V\setminus \mcM$. Let $\tau=(\MS(u))_{u\in\mcM}$, and let $\Phi_v^\tau$ denote the connected component of $\Phi^\tau$ that contains $v$.  Suppose $\theta:= 1 - \frac12 \exp \left(\frac{2 e d k }{2^{\alpha k}}\right)\geq 0.4$ and $d\leq 2^{\alpha k/4}$. Then with probability at least $1-n^{-0.1\log n}$ the number of clauses in $\Phi^\tau_v$ is at most $kd\log^2 n$.
\end{proposition}
\begin{proof}

Let $\Psi=(V_\Psi,\cC_\Psi)$ denote the induced $k$-CNF obtained by taking all the clauses in $\Phi$ in their entirety that appear (possibly in a reduced form) in $\Phi^\tau_v$. 
In particular $\Psi$ is $k$-uniform (although $\Phi^\tau_v$ might not be), and $\Phi_v^\tau$ and $\Psi$ have the same number of clauses. Let $\mathcal M_\Psi$ be the set of marked variables in $V_\Psi$. 

Recall the dependency hypergraph $H_\Phi=(V,\mathcal E)$ defined in \cref{def:dependency-graph}. Recall also from \cref{def:2tree} that $Z\subseteq V$ is a 2-tree if both of the following conditions are satisfied:
\begin{enumerate}
    \item for all $u,v\in Z$, $\text{dist}_{H_\Phi}(u,v)\geq 2$;
    \item the associated graph with vertex set $Z$ and edge set $\{\{u,v\}\subseteq Z: \dist_{H_\Phi}(u,v)=2\}$ is connected.
\end{enumerate}

Suppose $\Psi$ has $b$ clauses. Then by a greedy removal argument, one can find a 2-tree $Z\subseteq \mathcal M_\Psi$ with $|Z|= \floor{\frac{b}{kd}}$. For every $w\in Z$, let $C_w\in C_\Psi$ be any clause that contains $w$. We know that clauses in $\{C_w:w\in Z\}$ have pairwise disjoint variable sets as $\text{dist}_{H_\Phi}(u,v)\geq 2$ for all $u,v\in Z$. 

Since $\Phi_v^\tau$ is a single connected component, we know that for every marked variable $w'\in \text{vbl}(C_w)$, $\tau(w')$ does not satisfy $C_w$. This occurs with probability at most $1-\theta$, as we know that $\nu_{\mcM}$ is $\theta$-lower bounded, and $\MS(w')$ returns each of 0 and 1 with probability at least $\theta$.

Since every $C_w$ has at least $\alpha k$ marked variables,
we get that
$$
\PP[\text{$C_w$ is not satisfied by $(\tau(w'))_{w'\in\text{vbl}(C_w)\cap \mathcal M}$}]\leq (1-\theta)^{\alpha k}.
$$
Since the $C_w$ have disjoint variable sets, the events whether $C_w$ are satisfied by $(\tau(w'))_{w'\in\text{vbl}(C_w)\cap \mathcal M}$ are independent of each other. Thus, we know that 
$$
\PP[\text{none of the $C_w$, $w\in Z$ is  satisfied by $\tau$}]\leq (1-\theta)^{\alpha k|Z|}=(1-\theta)^{\alpha k\floor{\frac{b}{kd}}}.
$$

Finally, we union bound over all 2-trees of size $\floor{\frac{b}{kd}}$ that contains $v$. Recall from \cref{obs:2tree} that, since $H_\Phi$ is a hypergraph with maximum degree $d$, the number of 2-trees  $Z\subseteq H_\Phi$ containing $v$ with $|Z|=\ell$ is at most $\frac{(ed^2)^{\ell-1}}{2}$. Setting $\ell=\floor{\frac{b}{kd}}$ and applying a union bound, we get that
$$
\PP[\text{$\Phi_v^\tau$ has $b$ clauses}]=\PP[\text{$\Psi$ has $b$ clauses}]\leq \frac{(ed^2)^{\floor{\frac{b}{kd}}-1}}{2}(1-\theta)^{\alpha k\floor{\frac{b}{kd}}}\leq 0.9^{\floor{\frac{b}{kd}}}.
$$
Here in the last inequality, we used the fact that $\theta>0.4$ and the assumption that $d\leq 2^{\alpha k/4}$. Setting $b=kd\log^2 n$, we get that $\PP[\text{$\Phi_v^\tau$ has $kd\log^2 n$ clauses}]\leq 0.9^{\log^2 n}\leq n^{-0.1\log n}.$
\end{proof}

Observe that in $\CONN(\cdot)$, every clause is added to $\mathcal S$ because its variable set intersects with that of some unsatisfied clause in $\Phi_v^\tau$. Thus by \cref{prop:conn-size}, we know that with probability $1-n^{-0.1\log n}$ the number of clauses ever added to $\mathcal S$ is at most $\Delta kd\log^2n\leq k^2d^2\log^2n$. In this case, Line 7 of \cref{alg:conn} is executed at most $k^2d^2\log^2n\cdot k=k^3d^2\log^2n$ times. Combined with \cref{lem:time-gb}, we get that  the expected cost of $\CONN$ is at most
$$(k^3d^2\log^2n+n^{-0.1\log n}\cdot n)\cdot O(k^{17}d^{10}\log^2n/\alpha)=\poly\log(n). $$
Thus, \cref{t:conn} is proved.

\section{Proof of \cref{t:ismarked}}\label{app:marking}

Let $c>0$ be any constant. Given a $k$-uniform hypergraph $H$ and constants $0<\alpha<\beta_1<\beta_2<1$, we describe a local computation algorithm $\mathcal A$ that with probability $1-n^{-c}$ finds a 2-coloring of $H$, such that every edge in $H$ has at least $\alpha k$ vertices in each color. Taking $H=H_\Phi$ (as in \cref{def:dependency-graph}), we see that finding an $\alpha$-marking for $\Phi$ is equivalent to finding such a 2-coloring for $H_\Phi$. The algorithm $\mathcal A$  runs in three phases.
\begin{itemize}
    \item In Phase 1, we randomly color each vertex following an arbitrary order of the vertices.
    \begin{itemize}
        \item If an edge  at any point contains at least $\alpha k$ vertices in both colors, then we delete it from $H$.
        \item If an edge at any point contains more than $\beta_1 k$ vertices in one color,  then we call it a dangerous-1 edge and mark all its uncolored vertices as troubled-1. These troubled-1 vertices will not be colored in Phase 1.
        \item If an edge  at any point has all its vertices colored or troubled-1, but the edge itself is not deleted or dangerous-1, then we call this edge unsafe-1.
    \end{itemize}
    We proceed until all vertices in $H$ become colored or troubled-1.
    \item In Phase 2, we run the previous phase again on the remaining hypergraph $H$.
    There will be two differences:
    \begin{itemize}
        \item We do not recolor vertices that are already colored in Phase 1. In other words, we only randomly color those troubled-1 vertices.
        \item This time, we call an edge dangerous-2 if it contains more than $\beta_2 k$ vertices in one color. Again, if an edge becomes dangerous-2 at any point, then we mark all its  uncolored vertices as troubled-2.  These troubled-2 vertices will not be colored in Phase 2.
    \end{itemize}
    We proceed until all vertices in $H$ become colored or troubled-2.
    \item In Phase 3, we use an exhaustive search to color all the remaining (i.e., troubled-2) vertices in $H$, such that every remaining edge in $H$ has at least $\alpha k$ vertices in each color.
\end{itemize}

Suppose $1/2<\beta_1<\beta_2<1-\alpha$ satisfy conditions \cref{e:cond-1}:
\begin{align}
&4\alpha<2(1-\beta_2)<1-\beta_1, \notag\\
   & 16k^4d^5\leq 2^{(\beta_1-h(1-\beta_1))k}, \notag\\
    &16k^4d^5\leq 2^{(\beta_2-\beta_1)k-h\left(\frac{\beta_2-\beta_1}{1-\beta_1}\right)(1-\beta_1)k},\notag\\
    &\delta\mapsto  (\beta_2-\delta)-h\left(\frac{\beta_2-\delta}{1-\delta}\right)(1-\delta)\text{ is decreasing on $0\leq \delta\leq \beta_1$,} \notag\\
    &2e(kd+1)< 2^{\left(1-h\left(\frac{\alpha}{1-\beta_2}\right)\right)(1-\beta_2)k}. \notag
\end{align}
Then, the following holds regarding the above algorithm.

\begin{lemma}\label{l:phase-1}
    There exists $c_1>0$ sufficiently large such that with probability at least $1-n^{-2c}$, after Phase 1, all connected components in $H$ have size at most $c_1\log n$.
\end{lemma}
\begin{proof}
Let $H^3$ deote the distance-3 graph of $H$; that is, we have $V(H^3)=V(H)$, and $u\sim v$ in $H^3$ if they have distance 3 in $H$.  We say that $W\subseteq V(H)$ is a 3-tree if $H^3[W]$ is connected. 

Observe that for every $w>0$, there are at most $n4^w(dk)^{3w}$ 3-trees on $w$ vertices. This is because $H^3$ has maximum degree at most $(dk)^3$. Since there are at most $4^w$ non-isomorphic trees on $w$ vertices, and at most $n(dk)^{3w}$ copies of every such tree in $H^3$, we get that there are at most $n4^w(dk)^{3w}$ 3-trees on $w$ vertices.

Let $e$ be any edge, and $v$ be any vertex in $H$. Observe that
\begin{align*}
    \PP[\text{$e$ becomes dangerous-1 in Phase 1}]&\leq 2\binom{k}{\beta_1k}2^{-\beta_1k},\\
    \PP[\text{$e$ becomes unsafe-1 in Phase 1}]&\leq 2kd\binom{k}{\beta_1k}2^{-\beta_1k},\\
    \PP[\text{$v$ lies in a dangerous-1 or unsafe-1}&\text{ edge after Phase 1}]\\
    \qquad &\leq 4kd^2\binom{k}{\beta_1k}2^{-\beta_1k}\leq 4kd^22^{(h(1-\beta_1)-\beta_1)k}.
\end{align*}

Observe that every vertex that remains after Phase 1 lies in some dangerous-1 or unsafe-1 edge. Let $W$ be a 3-tree on $w$ vertices, and $p=4kd^22^{(h(1-\beta_1)-\beta_1)k}$. Since no two vertices in $W$ lie in edges that intersect, we have 
$$\PP[\text{$W$ is contained in some connected component after Phase 1}]\leq p^{w}.$$
Thus, the expected number of 3-trees on $w$ vertices after Phase 1 is at most
$$n4^w(dk)^{3w}p^w=n(4d^3k^3\cdot 4kd^22^{(h(1-\beta_1)-\beta_1)k})^w=n(16k^4d^5\cdot 2^{(h(1-\beta_1)-\beta_1)k})^w,$$
which is at most $n^{-2c}$ when $w=c_1\log n$ with $c_1>0$ sufficiently large. By Markov's inequality, with probability at least $1-n^{-2c}$, there is no 3-tree of size larger than $c_1 \log n$. Since every size-$s$ connected component in $H$ contains a 3-tree of size at least $s/(dk)^3$, we get that with probability at least $1-n^{-2c}$, there is no connected component of size larger than than $c_1(dk)^3 \log n$ in $H$.
\end{proof}

\begin{lemma}\label{l:phase-2}
    Let $H_1$ be a connected component of $H$ after Phase 1. Suppose $H_1$ has size $N\leq c_1\log n$ with $c_1>1$. Then there exists $c_2>0$ sufficiently large such that with probability at least $1-1/\log^2n$, after Phase 2, all connected components in $H_1$ have size at most $c_2\log \log n$.
\end{lemma}
\begin{proof}Let $\widetilde H_1$ be the hypergraph obtained by deleting all colored vertices from $H_1$. Suppose $N\leq c_1\log n$ is the vertex number of $\widetilde H_1$. Note that $\widetilde H_1$ also has maximum degree at most $d$, and $\widetilde H_1^3$ also has maximum degree at most $(dk)^3$.

The argument in similar to the proof of \cref{l:phase-1}. Let $e$ be any edge, and $v$ be any vertex in $\widetilde H_1$. Observe that
\begin{align*}
    &\PP[\text{$e$ becomes dangerous-2 in Phase 2}]\leq 2\binom{(1-\delta)k}{(\beta_2-\delta)k}2^{-(\beta_2-\delta)k}\text{ for some $0\leq \delta\leq\beta_1$},\\
    &\PP[\text{$e$ becomes unsafe-2 in Phase 2}]\leq 2kd\binom{(1-\delta)k}{(\beta_2-\delta)k}2^{-(\beta_2-\delta)k},\\
    &\PP[\text{$v$ lies in a dangerous-2 or unsafe-2 edge after Phase 1}]\\
    &\qquad \leq 4kd^2\binom{(1-\delta)k}{(\beta_2-\delta)k}2^{-(\beta_2-\delta)k}\\
    &\qquad\leq 4kd^22^{h\left(\frac{\beta_2-\delta}{1-\delta}\right)(1-\delta)k-(\beta_2-\delta)k}\leq 4kd^22^{h\left(\frac{\beta_2-\beta_1}{1-\beta_1}\right)(1-\beta_1)k-(\beta_2-\beta_1)k}.
\end{align*}
Let $\widetilde p=4kd^22^{h\left(\frac{\beta_2-\beta_1}{1-\beta_1}\right)(1-\beta_1)k-(\beta_2-\beta_1)k}$. Thus, 
the expected number of 3-trees on $w$ vertices after Phase 1 is at most
\begin{align*}
    N4^w(dk)^{3w}\widetilde p^w&=N(4d^3k^3\cdot 4kd^22^{h\left(\frac{\beta_2-\beta_1}{1-\beta_1}\right)(1-\beta_1)k-(\beta_2-\beta_1)k})^w\\
    &=N(16k^4d^5\cdot 2^{h\left(\frac{\beta_2-\beta_1}{1-\beta_1}\right)(1-\beta_1)k-(\beta_2-\beta_1)k})^w,
\end{align*}
which is at most $N^2\leq c_1^2\log^2n$ when $w=c_2\log N$ with $c_2>0$ sufficiently large. By Markov's inequality, with probability at least $1-1/(c_1^2\log^2n)\geq 1-1/\log^2n$, there is no 3-tree of size larger than $c_2 \log N\leq c_2\log(c_1\log n)$. This implies that with probability at least $1-1/\log^2n$, there is no connected component of size larger than than $c_2(dk)^3 \log(c_1\log n)\leq c_1c_2(dk)^3 \log\log n$ in $\widetilde H_1$.
\end{proof}

\begin{lemma}\label{l:phase-3}
Let $H_2$ be a connected component of $H$ after Phase 2. Suppose $H_2$ has size at most $c_2\log \log n$. Then there exists  a valid coloring of $H_2$, i.e., a coloring in which every edge has at least $\alpha n$ vertices under each color.
\end{lemma}
\begin{proof}
Consider a random 2-vertex coloring of $H_2$. For every edge $e$ in $H_2$, let $A_e$ be the event that $e$ does not have $\alpha n$ vertices under one of the colors. Then we have
$$\PP[A_e]\leq 2\binom{(1-\beta_2)k}{\leq \alpha k}2^{-(1-\beta_2)k}\leq 2^{\left(h\left(\frac{\alpha}{1-\beta_2}\right)-1\right)(1-\beta_2)k+1}.$$
Since each $A_e$ is independent from all but at most $kd$ other $A_{e'}$'s, with
$$e\cdot 2^{\left(h\left(\frac{\alpha}{1-\beta_2}\right)-1\right)(1-\beta_2)k+1} \cdot (kd+1)<1,$$
we know from the Lov\'asz local lemma that with positive probability, none of the $A_e$'s occurs.
\end{proof}

We now argue that for every $v\in V(H)$, we can locally access the color of $v$ given by the above algorithm. To do so, observe that the above algorithm works for an arbitrary variable ordering. In particular, a random ordering will work (which we can locally look up  given a public tape of random bits $\mathbf{R}$).

We associate a random value $a_v\in[0,1]$ to every vertex $v\in V(H)$, and order the vertices as $v_1,\dots,v_n$ such that $a_{v_1}<\dots<a_{v_n}$.  For every vertex $v\in V$, we compute the query tree $\mathcal T_v$ defined as follows:
\begin{itemize}
    \item Let $v$ be the root of $\mathcal T_v$.
    \item For every $w\in V$ that shares a clause with $v$, if $a_w<a_v$, then we add $w$  to $\mathcal T_v$ as a child of $v$.
    \item We build  $\mathcal T_v$ in this way recursively until  no child node can be added to $\mathcal T_v$.
\end{itemize}
We know that with high probability, $\mathcal T_v$ has $\mathrm{polylog}(n)$ vertices.

\begin{lemma}[\text{\cite[Theorem 3.1]{ARV12}}]\label{lem:query-tree}
    Fix $c>1$. For every $k,d\geq 1$, there exists  $c_3>0$ such that for all sufficiently large $n$, $\PP[|\mathcal T_v|>c_3\log^{kd+1}(n)]<n^{-2c}$.
\end{lemma}

In Phase 1, whether $v$ is colored or troubled-1 completely depends on its descendants in $\mathcal T_v$. Thus, to compute the assignment of $v$ in Phase 1, it suffices to run Phase 1 on all vertices in $\mathcal T_v$. By \cref{lem:query-tree}, with probability $1-n^{-2c}$, this step takes $\mathrm{polylog}(n)$ time.

If $v$ gets colored in Phase 1, then we return the color directly. Otherwise, if $v$ becomes troubled-1 after Phase 1, then we compute the connected component of $H$ containing $v$ after Phase 1. Recall from \cref{l:phase-1} that with probability $1-n^{-2c}$, all connected components after Phase 1 has size at most $c_1\log n$, so this step takes another $\mathrm{polylog}(n)$ time.

Let $H_v$ denote this connected component. By \cref{l:phase-2}, if we perform  $O(\log n/\log\log n)$ parallel instances of Phase 2 on $H_v$,  then with probability $1-n^{-2c}$, in one of them $H_v$ is further reduced to connected components each of size $c_2\log\log n$. We then perform Phase 3 to find a valid coloring on the component that contains $v$, which by \cref{l:phase-3} is guaranteed to succeed and finishes in $2^{O(\log\log n)}=\mathrm{polylog}(n)$ time.

Combining the above analysis, we have obtained the local computation algorithm $\mathsf{IsMarked}(\cdot)$  with the desired properties.

\section{Proof of~\cref{lem:constants}}\label{a:verify-consts}
    We first check conditions in \cref{e:cond-1}. For sufficiently large $k$, we have:
    \begin{itemize}
        \item $4\cdot \frac{1}{75}<2(1-0.96)<1-0.778$,
        \item  $16k^4d^5\leq  16k^42^{k/80}\leq 2^{0.014k}\leq 2^{(\beta_1-h(1-\beta_1))k}$,
        \item  $16k^4d^5\leq 16k^42^{k/80}\leq 2^{0.03k}\leq 2^{(\beta_2-\beta_1)k-h\left(\frac{\beta_2-\beta_1}{1-\beta_1}\right)(1-\beta_1)k}$,
        \item $2e(kd+1)< 2e(k2^{k/400}+1)\leq 2^{0.003k}\leq  2^{\left(1-h\left(\frac{\alpha}{1-\beta_2}\right)\right)(1-\beta_2)k}$.
    \end{itemize}
    Moreover, letting $f(\delta):=(\beta_2-\delta)-h\left(\frac{\beta_2-\delta}{1-\delta}\right)(1-\delta)$, we have
\begin{align*}
    f'(\delta)&=-\frac{\left(1-\frac{0.96\, -\delta}{1-\delta}\right) \log \left(1-\frac{0.96\,
   -\delta}{1-\delta}\right)}{\log (2)}-\frac{(0.96\, -\delta) \log \left(\frac{0.96\,
   -\delta}{1-\delta}\right)}{(1-\delta) \log (2)}\\
   &\qquad -(1-\delta) \bigg(-\frac{\frac{0.96\,
   -\delta}{(1-\delta)^2}-\frac{1}{1-\delta}}{\log (2)}-\frac{\left(\frac{1}{1-\delta}-\frac{0.96\,
   -x}{(1-x)^2}\right) \log \left(1-\frac{0.96\, -\delta}{1-\delta}\right)}{\log (2)}\\
   &\qquad\qquad  +\frac{\log
   \left(\frac{0.96\, -\delta}{1-\delta}\right)}{(1-x) \log (2)}-\frac{(0.96\, -x) \log
   \left(\frac{0.96\, -\delta}{1-\delta}\right)}{(1-\delta)^2 \log (2)}-\frac{\frac{1}{1-\delta}-\frac{0.96\,
   -\delta}{(1-\delta)^2}}{\log (2)}\bigg)-1<0
\end{align*}
for $0\leq \delta\leq \beta_1=0.778$.

We now check conditions in \cref{e:cond-2}. For sufficiently large $k$, we have:
\begin{itemize}
    \item $k \cdot 2^{-\alpha k} \cdot (dk) ^5\cdot 4\leq k^62^{k/80-k/75} \cdot 4\leq \frac{1}{150e^3}$,
    \item $\theta=1 - \frac12 \exp \left(\frac{2 e d k }{2^{\alpha k}}\right)\geq 1-\frac{1}{2}\exp(2ek2^{k/400}-2^{k/75})\geq 0.4$,
    \item $ 36 e d^3k^4 \cdot 0.6^{\alpha k}=36ek^4(2^{3/400}\cdot 0.6^{1/75})^k\leq 1/2$,
    \item $ 2^{-\frac{1}{48d k^4}}\cdot e^{\frac{2d^2 /\alpha}{2^{\alpha k}}}\leq  2^{-\frac{1}{48k^42^{k/400}}}\cdot e^{\frac{150\cdot 2^{k/200}}{2^{k/75}}}=2^{\log_2\left(-\frac{1}{48k^42^{k/400}}+ \frac{150\log_2e\cdot 2^{k/200}}{2^{k/75}}\right)}\leq   0.9$,
    \item $d\leq 2^{k/400}\leq 2^{\alpha k/4}.$
\end{itemize}

\end{document}